\documentclass{article}

\usepackage{csquotes}
\usepackage[a4paper,margin=2cm]{geometry}
\usepackage{mathtools}
\usepackage{amssymb}
\usepackage{amsthm}
\usepackage{thmtools}
\usepackage{xspace}
\usepackage[dvipsnames]{xcolor}
\usepackage{bbm}
\usepackage[colorlinks=true, allcolors=Maroon]{hyperref}
\usepackage[capitalise,nameinlink]{cleveref}
\usepackage{physics}
\usepackage[framemethod=tikz]{mdframed}
\usepackage{dirtytalk}
\usepackage{tikz}
\usetikzlibrary{arrows.meta,calc,positioning}

\newtheorem{theorem}{Theorem}[section]
\newtheorem{definition}[theorem]{Definition}
\newtheorem{lemma}[theorem]{Lemma}
\newtheorem{corollary}[theorem]{Corollary}
\newtheorem{observation}[theorem]{Observation}
\newtheorem{problem}[theorem]{Problem}
\newtheorem{proposition}[theorem]{Proposition}
\newtheorem{result}{Result}

\newcommand{\classname}[1]{\ensuremath{\mathsf{#1}}\xspace} 
\newcommand{\BQP}{\classname{BQP}}
\newcommand{\BPP}{\classname{BPP}}
\newcommand{\PP}{\classname{PP}}
\newcommand{\PH}{\classname{PH}}

\newcommand{\loosesym}{\uparrow} % marker for loose oracle access
\newcommand{\LooseUP}{\classname{PromiseUP}_{\loosesym}}
\newcommand{\LooseUSigma}[1]{\classname{PromiseU\Sigma}_{#1,\loosesym}}
\newcommand{\LooseUH}{\classname{PromiseUH}_\loosesym}

\newcommand{\PromiseBQP}{\classname{PromiseBQP}}
\newcommand{\PromiseBPP}{\classname{PromiseBPP}}
\newcommand{\LooseBPP}{\classname{PromiseBPP}_{\loosesym}}
\newcommand{\BPdot}{\ensuremath{\mathsf{BP}\cdot}}
\newcommand{\QMA}{\classname{QMA}}
\newcommand{\QCMA}{\classname{QCMA}}
\newcommand{\QCPH}{\classname{QCPH}}
\newcommand{\QCSigma}[1]{\classname{QC\Sigma_{#1}}}
\newcommand{\comp}{\mathrm{Comp}} %Completion
\newcommand{\GapP}{\classname{GapP}}
\newcommand{\SPP}{\classname{SPP}}
\newcommand{\SBP}{\classname{SBP}}
\newcommand{\bqpqpoly}{\classname{PromiseBQP}_{\mathrm{/qpoly}}}
\newcommand{\AWPP}{\classname{AWPP}}
\newcommand{\PromiseAWPP}{\classname{PromiseAWPP}}
\newcommand{\PromiseAPP}{\classname{PromiseAPP}}
\newcommand{\PSPACE}{\classname{PSPACE}}
\newcommand{\Pclass}{\classname{P}}
\newcommand{\PromiseP}{\classname{PromiseP}}
\newcommand{\PPoly}{\classname{PromiseP}_{\mathrm{/poly}}}
\newcommand{\PostBQPstar}{\classname{PromisePostBQP^*}}
\newcommand{\PromiseYQPstar}{\classname{PromiseYQP^*}}
\newcommand{\YQPstar}{\classname{YQP^*}}
\newcommand{\CH}{\classname{CH}}
\newcommand{\CHi}[1]{\classname{C_{#1}P}}
\newcommand{\PostBQP}{\classname{PostBQP}}
\newcommand{\NP}{\classname{NP}}
\newcommand{\coNP}{\classname{coNP}}
\newcommand{\APP}{\classname{APP}}

\newcommand{\yes}{\mathrm{yes}}
\newcommand{\no}{\mathrm{no}}
\newcommand{\qthr}{\ensuremath{\mathrm{QThr}}\xspace}
\newcommand{\thr}{\ensuremath{\mathrm{Thr}}\xspace}
\newcommand{\qacc}{\ensuremath{\mathrm{QAcc}}\xspace}
\newcommand{\qcirc}{\ensuremath{\mathrm{QCirc}}\xspace}
\newcommand{\tqbf}{\ensuremath{\mathrm{TQBF}}\xspace}
\newcommand{\calD}{\ensuremath{\mathcal{D}}\xspace}
\newcommand{\simyes}{\calD_{n,\yes}}
\newcommand{\simno}{\calD_{n,\no}}
\newcommand{\addr}{\ensuremath{\mathrm{addr}}\xspace}
\newcommand{\wt}{\ensuremath{\mathrm{wt}}\xspace}

\newcommand{\calC}{\ensuremath{\mathcal{C}}}
\newcommand{\calQ}{\ensuremath{\mathcal{Q}}}
\newcommand{\calO}{\ensuremath{\mathcal{O}}}
\newcommand{\calA}{\ensuremath{\mathcal{A}}}
\newcommand{\calB}{\ensuremath{\mathcal{B}}}
\newcommand{\calT}{\ensuremath{\mathcal{T}}}
\newcommand{\calM}{\ensuremath{\mathcal{M}}}
\DeclareMathOperator{\E}{\mathbb{E}}
\newcommand{\poly}{\ensuremath{\mathrm{poly}}}
\newcommand{\dom}{\ensuremath{\mathrm{Dom}}}
\DeclareMathOperator{\Id}{\mathbbm{1}}
\DeclarePairedDelimiter\floor{\lfloor}{\rfloor}
\DeclarePairedDelimiter\ceil{\lceil}{\rceil}

\newmdenv[
  topline=false,
  bottomline=false,
  rightline=false,
  linewidth=3pt,
  linecolor=black!20,
  innerleftmargin=10pt,
  innerrightmargin=0pt,
  innertopmargin=2pt,
  innerbottommargin=2pt,
  skipabove=6pt,
  skipbelow=0pt,
]{algbox}

\title{Promises should be taken seriously:\\ On relativization with promise problems}
\author{David Miloschewsky\thanks{Department of Computer Science, Stony Brook University, USA.  \\ Email: \texttt{\{dmiloschewsk, supartha\}@cs.stonybrook.edu}.} \and Supartha Podder\footnotemark[1] \and  Dorian Rudolph\thanks{Department of Computer Science and Institute for Photonic Quantum Systems (PhoQS),
Paderborn University, Germany. Email: \texttt{dorian.rudolph@upb.de}.}}
\date{}

\begin{document}

\maketitle

\begin{abstract}
    Relativization is concerned with comparing computational models with black-box access to an oracle. For promise problems, black-box access is not canonical due to inputs outside of the promise being unconstrained. We study two semantics for such access. Under \emph{robust} queries, a machine must correctly answer regardless of the completion of the problem,, while \emph{loose} access requires that the internal choices of a machine do not change based on off-promise queries. We use both notions to study the consequences of the choice of model.

    Our first result separates the language and promise settings. Namely, we construct an oracle $O$ such that,
    \begin{align*}
        \mathsf{P}^O = \mathsf{BQP}^O = \mathsf{AWPP}^O \text{, but } \mathsf{PromiseBQP}^O \not\subseteq \mathsf{PromiseP}^O_{\mathsf{/poly}}.
    \end{align*}
    In particular, $\mathsf{BPP}^O = \mathsf{BQP}^O$, but $\mathsf{PromiseBQP}^O \neq \mathsf{PromiseBPP}^O$, showing that results for languages need not transfer to promises. We show this by embedding an instance of Simon's problem into an oracle which is Cohen-generic relative to $\mathsf{PSPACE}$ oracle.
    Next, we use loose queries to strengthen the upper bound on the Quantum-Classical Polynomial Hierarchy from $\mathsf{P}^{\mathsf{PP}^{\mathsf{PP}}}$ to the following,
    \begin{align*}
        \mathsf{QCPH} \subseteq \mathsf{BP\cdot PP} \subseteq \mathsf{PromiseBPP}^{\mathsf{PP}}.
    \end{align*}
    The same proof also shows $\mathsf{PP}^\mathsf{PromiseBQP} = \mathsf{PP}$. Additionally, we show that $\mathsf{PromiseBQP}$, even when given quantum advice, is self-low under robust queries.

    Finally, we exhibit an obstruction to transferring language-level counting results to promise classes. Although $\mathsf{AWPP}$ and $\mathsf{APP}$ are low for $\mathsf{PP}$, a corresponding promise analogue would collapse the counting hierarchy as $\mathsf{GapP} \subseteq \mathsf{FP}^{\mathsf{PromiseAWPP}}$. This motivates the introduction of $\mathsf{PromisePostBQP^*}$, which restricts $\mathsf{PostBQP}$ to input-indepencent postselection. By showing that it is low for \PP, we obtain $\mathsf{PP}^{\mathsf{PromiseYQP^*}} = \mathsf{PP}$. Hence our results show that for semantic complexity classes, the treatment of off-promise queries is essential.
\end{abstract}

\section{Introduction}

The study of complexity classes is usually formulated in terms of languages, sets which represent total Boolean functions. Relativization traditionally equips machines with a language oracle, so that every query has an answer. This convention underlies both the classical and standard quantum oracle models~\cite{BGS75,BV97}. Promise problems generalize languages by specifying two disjoint sets of yes and no instances whose union is called the promise, but say nothing about the remaining off-promise inputs. Even, Selman, and Yacobi introduced promise problems and immediately recognized that the same relationships which hold for languages need not hold for promise problems~\cite{ESY84}. 

In particular, they show that the existence of an $\NP$-hard promise problem $\Pi \in \classname{PromiseNP}\cap \classname{PromisecoNP}$ does not necessarily imply that $\NP = \coNP$, as a reduction may query $\Pi$ outside of its promise. On the other hand, if the reduction is forbidden from querying off-promise, the standard argument goes through. However, this means that the reduction must guarantee that every query lies inside a promise, even if membership may be difficult to decide. To complicate matters further, when considering the quantum setting, an algorithm may place promise and off-promise queries in the same superposition~\cite{BV97,NM01}. Thus, one must carefully decide the exact conditions under which promise problems may be queried.

As discussed in~\cite{Gol06}, promise problems naturally capture the power of semantic complexity classes. For example, a \PromiseBQP machine is able to distinguish inputs whose acceptance probability is at least $\tfrac{2}{3}$ from those with acceptance probability at most $\tfrac{1}{3}$, but is allowed to answer arbitrarily on inputs in the gap. Let $\Pi = (\Pi_\yes, \Pi_\no)$ be such a promise problem. We say that a language $L$ is a completion of $\Pi$ if it contains $\Pi_\yes$ and is disjoint from $\Pi_\no$. One may be tempted to argue that for any promise problem, we may construct a completion that preserves the computational power of the problem. However, the following relativized example shows that need not be true.

\begin{observation}\label{obs:no_valid_completions}
    There exists an oracle $O$ such that, for every $\PromiseBQP^O$-complete promise problem $\Pi$ and every completion $A$ of $\Pi$, $A \notin \BQP^O$.
\end{observation}
\begin{proof}
    Fortnow and Rogers show that there exists an oracle $O$ with respect to which $\BQP^O$ does not have a complete language~\cite{FR99}. Let $\Pi$ be a $\PromiseBQP^O$-complete promise problem (such as evaluating a quantum circuit with access to $O$). Consider some arbitrary language $L\in \BQP^O$. Then $(L, \{0,1\}^* \setminus L)\in \PromiseBQP^O$, so completeness of $\Pi$ gives a polynomial-time map $f$ with access to $O$ from $L$ to $\Pi$ whose output always lies in $\Pi_\yes \cup \Pi_\no$. Choose any completion $A$ of $\Pi$. Because $A$ agrees with $\Pi$ on the promise, $x\in L \iff f(x)\in A$. If $A\in \BQP^O$, this would make $A$ a $\BQP^O$-complete language, contradicting the choice of oracle $O$.
\end{proof}

Consequently, choosing a completion and assuming it lies in the corresponding language class is not a sound argument. Hence, when considering queries to promise problems, one cannot simply fix a completion without specifying how off-promise queries are treated.

\begin{figure}
    \centering
    \begin{tikzpicture}[
        x=1cm,
        y=1cm,
        font=\small,
        panel/.style={draw=black!38, rounded corners=2pt, line width=.55pt},
        oracle/.style={minimum height=.58cm, inner sep=0pt, font=\scriptsize},
        machine/.style={
            draw=black!55,
            rounded corners=2pt,
            minimum width=1.05cm,
            minimum height=.62cm,
            fill=black!2,
            font=\footnotesize
        },
        query/.style={
            -{Latex[length=1.8mm,width=1.1mm]},
            line width=.65pt
        },
        good/.style={
            draw=ForestGreen!65!black,
            fill=ForestGreen!15,
            line width=.45pt
        },
        bad/.style={
            draw=BrickRed!55!black,
            fill=BrickRed!9,
            line width=.45pt
        },
        celltext/.style={font=\scriptsize},
        headline/.style={font=\bfseries, anchor=north west}
    ]

        % More compact panel frames.
        \foreach \x in {0,4.97,9.94}
            \draw[panel] (\x,0) rectangle ++(4.72,3.75);

        % ------------------------------------------------------------
        % Smart
        % ------------------------------------------------------------
        \begin{scope}[shift={(0,0)}]
            \node[headline] at (.22,3.53) {Smart};

            \node[
                oracle,
                draw=ForestGreen!65!black,
                fill=ForestGreen!14,
                minimum width=1.20cm
            ] (syes) at (.92,2.50) {$\Pi_{\yes}$};

            \node[
                oracle,
                draw=black!45,
                dashed,
                fill=black!5,
                minimum width=1.62cm
            ] (soff) at (2.36,2.50) {off-promise};

            \node[
                oracle,
                draw=BrickRed!65!black,
                fill=BrickRed!11,
                minimum width=1.20cm
            ] (sno) at (3.80,2.50) {$\Pi_{\no}$};

            \node[machine] (sm) at (2.36,.82) {$M$};

            \draw[query]
                (sm.north west) to[out=135,in=-90] (syes.south);

            \draw[query]
                (sm.north east) to[out=45,in=-90] (sno.south);

            \draw[query, dashed, black!35]
                (sm.north) -- (soff.south);

            \draw[BrickRed!85!black, line width=1.15pt]
                ($(sm.north)!0.56!(soff.south)+(-.18,-.18)$) --
                ($(sm.north)!0.56!(soff.south)+(.18,.18)$);

            \draw[BrickRed!85!black, line width=1.15pt]
                ($(sm.north)!0.56!(soff.south)+(-.18,.18)$) --
                ($(sm.north)!0.56!(soff.south)+(.18,-.18)$);
        \end{scope}

        % ------------------------------------------------------------
        % Robust access
        % ------------------------------------------------------------
        \begin{scope}[shift={(4.97,0)}]
            \node[headline] at (.22,3.53) {Robust};

            \node[font=\scriptsize, text=black!65]
                at (2.44,3.02) {completion};

            \node[font=\scriptsize, rotate=90, text=black!65]
                at (.42,1.74) {random string};

            \foreach \j/\lab in {0/A_1,1/A_2,2/A_3}
                \node[font=\scriptsize]
                    at (1.72+.72*\j,2.77) {$\lab$};

            \foreach \i/\lab in {0/r_1,1/r_2,2/r_3}
                \node[font=\scriptsize, anchor=east]
                    at (1.25,2.30-.56*\i) {$\lab$};

            % Each completion has a 2/3 fraction of successful strings,
            % but the successful strings depend on the completion.
            \foreach \i/\j in {
                0/0,0/1,
                1/0,1/2,
                2/1,2/2
            } {
                \draw[good]
                    (1.37+.72*\j,2.02-.56*\i)
                    rectangle ++(.70,.56);

                \node[celltext, text=ForestGreen!45!black]
                    at (1.72+.72*\j,2.30-.56*\i)
                    {$\checkmark$};
            }

            \foreach \i/\j in {0/2,1/1,2/0} {
                \draw[bad]
                    (1.37+.72*\j,2.02-.56*\i)
                    rectangle ++(.70,.56);

                \node[celltext, text=BrickRed!70!black]
                    at (1.72+.72*\j,2.30-.56*\i)
                    {$\times$};
            }

            \node[font=\scriptsize] at (2.36,.36)
                {$\forall A\;
                  \Pr_r[M^A\text{ succeeds}]
                  \geq \tfrac{2}{3}$};
        \end{scope}

        % ------------------------------------------------------------
        % Loose access
        % ------------------------------------------------------------
        \begin{scope}[shift={(9.94,0)}]
            \node[headline] at (.22,3.53) {Loose};

            \node[font=\scriptsize, text=black!65]
                at (2.44,3.02) {completion};

            \node[font=\scriptsize, rotate=90, text=black!65]
                at (.42,1.74) {random string};

            \foreach \j/\lab in {0/A_1,1/A_2,2/A_3}
                \node[font=\scriptsize]
                    at (1.72+.72*\j,2.77) {$\lab$};

            \foreach \i/\lab in {0/r_1,1/r_2,2/r_3}
                \node[font=\scriptsize, anchor=east]
                    at (1.25,2.30-.56*\i) {$\lab$};

            % One common 2/3 fraction succeeds for every completion.
            \draw[
                ForestGreen!55!black,
                fill=ForestGreen!4,
                rounded corners=1pt,
                line width=.75pt
            ]
                (1.33,1.43) rectangle ++(2.22,1.15);

            \foreach \i in {0,1} {
                \foreach \j in {0,1,2} {
                    \draw[good]
                        (1.37+.72*\j,2.02-.56*\i)
                        rectangle ++(.70,.56);

                    \node[celltext, text=ForestGreen!45!black]
                        at (1.72+.72*\j,2.30-.56*\i)
                        {$\checkmark$};
                }
            }

            \foreach \j in {0,1,2} {
                \draw[bad]
                    (1.37+.72*\j,.90)
                    rectangle ++(.70,.56);

                \node[celltext, text=BrickRed!70!black]
                    at (1.72+.72*\j,1.18)
                    {$\times$};
            }

            \node[
                font=\scriptsize,
                text=ForestGreen!45!black,
                anchor=west
            ] at (3.64,2.02) {$2/3$};

            \node[font=\scriptsize] at (2.36,.36)
                {$\Pr_r[
                    \forall A:\,
                    M^A\text{ succeeds}
                  ]\geq\tfrac23$};
        \end{scope}

\end{tikzpicture}
    \caption{Three different types of oracle access to promise problems.}\label{fig:visualize_query_types}
\end{figure}
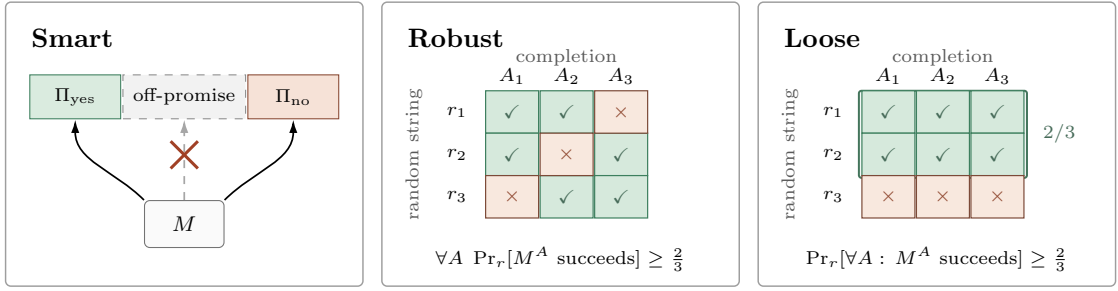

There are several alternative ways to define oracle access to promise problems. \emph{Smart} oracle access means that the machine may not query outside of the oracle's promise~\cite{GS88}, meaning that the machine is \say{smart} enough to avoid meaningless queries. However, this can be a strong requirement. First, for queries in superposition, this would forbid placing any amplitude on off-promise queries, even when the total weight on these queries is negligible and cannot affect the final result. Second, a classical reduction may not have an efficient way to recognize whether a query lies inside the promise. For promise problems defined by acceptance thresholds, a query close to the threshold may require resolving a fine-grained acceptance probability comparison. These issues motivate allowing for off-promise queries while requiring correctness for any completion.

Under \emph{robust} oracle access, we say that a single machine solves the problem regardless of the completion of the oracle promise problem.\footnote{This type of access was called uniform in~\cite{ESY84}. Its probabilistic analogue was considered in~\cite{BF99}.} The machine's behavior may depend on the completion and not all computational routes may end in the same result. \emph{Loose} access strengthens this requirement by switching the order of quantifiers. Once we fix some internal choice of the machine, such as a random string or a witness, the same choice must work for every completion~\cite{CR08,HS15}. The relevant internal choice depends on the model, so for each machine one must exactly specify what this means.

For example, consider a $\BPP^\Pi$ machine $M$ where $\Pi$ is a promise problem oracle and $x$ some yes-input. If $\Pi$ is given under loose access, we require that for at least $\tfrac{2}{3}$ of the random strings $r$, for any completion $A$ of $\Pi$, we have $M^A(x,r) = 1$. On the other hand, under robust access we require simply that for each completion, at least $\tfrac{2}{3}$ of the random strings accept, even if the intersection of the accepting random strings $r$ across all completions is empty. A visualization of the models may be found in \cref{fig:visualize_query_types}.

\subsection{Our results}

We study how various relativization results behave when one uses promise problems instead of languages. First, we show that relativizing results for languages need not transfer to promise problems.

\begin{result}[\cref{thm:promise_language_difference}]\label{res:oracle_separation}
    There exists an oracle $O$ such that,
    \begin{align*}
        \Pclass^O = \BQP^O = \AWPP^O \text{, but } \PromiseBQP^O \not\subseteq \PPoly^O.
    \end{align*}
\end{result}
As the standard result $\PromiseBPP \subseteq \PPoly$~\cite{Adl78} relativizes, this answers a question posed by Nisan~\cite{Aar10} asking whether there exists an oracle under which $\BPP= \BQP$ but $\PromiseBQP \neq \PromiseBPP$. This is interesting when compared to the result that $\classname{FBPP} = \classname{FBQP} \iff \classname{SampP} = \classname{SampBQP}$~\cite{Aar14} as well as the phenomenon in quantum query complexity where no total Boolean function can attain a superpolynomial quantum query complexity speedup over classical algorithms, but promise functions can~\cite{BBC+01,ABK16,BDC+20, PYY25,HKMP26}.

The oracle $O$ is constructed by combining a \PSPACE-complete language with a Cohen-generic set, a combination which was shown in~\cite{FFKL03} to collapse the language classes in~\cref{res:oracle_separation}. Specifically, we encode independent instances of Simon's problem~\cite{Sim97} into the Cohen-generic set and argue using the direct-product bound of~\cite{Dru12} that machines with advice cannot solve all instances simultaneously.

Next, we use the loose oracle access to improve the upper bound on \QCPH, the Quantum-Classical Polynomial Hierarchy. Toda's theorem shows that $\PH \subseteq \Pclass^\PP$~\cite{Tod91}. However, the same upper bound for the hierarchical generalization of \QCMA has been elusive. While various properties such as a collapse theorem, error-reduction, or level-wise oracle separations have been developed~\cite{AGKR24,ABD25}, the best known upper bound $\Pclass^{\PP^\PP}$ has not changed since the class was introduced in~\cite{GSS+22}.

\begin{result}[\cref{thm:qcph_bpp_pp}]\label{res:qcph}
    We improve the upper bound on \QCPH to
    \begin{align*}
        \QCPH \subseteq \BPdot\PP \subseteq \PromiseBPP^\PP.
    \end{align*}
\end{result}
Here, $\BPdot\PP$ denotes the class of promise problems with a randomized polynomial-time many-one reduction to \PP. The primary difficulty in improving the bound is the fact that \QCPH is a promise class whose quantified strings are evaluated using \PromiseBQP computation. As was shown in~\cref{obs:no_valid_completions}, we cannot simply replace the \PromiseBQP computation by an arbitrary completion assumed to be in \BQP. Thus, we instead adapt Toda's theorem to work directly with promise problems.

The proof is split into four parts. First, following the lineage of Valiant-Vazirani~\cite{VV86} and its quantum extension~\cite{ABOBS22}, we show that \QCPH reduces to randomized computation with loose access to the unambiguous hierarchy, which itself has loose access to \PromiseBQP, extending the work of~\cite{HS15}. Note that loose access to the unambiguous hierarchy means that the unique accepting witness is the same witness regardless of the completion. Second, we give an \SPP characterization of the loose unambiguous hierarchy. Third, we compress the random seed of the loose-access computation into the input of a single query, which turns the result into a randomized many-one reduction. 

Last, we remove the \PromiseBQP oracle by sampling an amplified \GapP representation and using the closure properties of \GapP to show that its expectation is in \PP. As an intermediate result, we also show that \PromiseBQP is low for \PP, paralleling the analogous result for languages~\cite{FR99}. We remark that the first three parts do not use any property of quantum computation. Instead, we prove them for the hierarchy obtained by alternately maximizing and minimizing an arbitrary function $f:\{0,1\}^*\to [0,1]$, given loose access to a gapped threshold oracle for $f$, and obtain the bound on \QCPH as a corollary.

Continuing with queries to \PromiseBQP, we notice that the proof of self-lowness of \BQP~\cite{BBBV97} does not immediately extend to \PromiseBQP. The issue is that any simulation of \PromiseBQP may have acceptance probability extremely close to $\tfrac{1}{2}$ on off-promise queries, even after boosting. We overcome this by randomizing the threshold used to simulate each query.
\begin{result}[\cref{thm:promise_bqp_selflow,cor:promise_bqpqpoly_selflow}]
    Under robust queries,
    \begin{align*}
        &\PromiseBQP^\PromiseBQP = \PromiseBQP\text{, and} &\bqpqpoly^{\bqpqpoly} = \bqpqpoly.
    \end{align*}
\end{result}
The proof technique is essentially a quantization of the approach for \PromiseBPP~\cite{Gol11}. We partition the off-promise completeness-soundness interval into $\poly(n)$ sub-intervals and show that if we randomly assign a sub-interval to each distinct query, then unless the selected sub-interval contains the query's acceptance probability, we can determine whether the probability lies above or below it while agreeing with the promise queries. As there is at most one bad sub-interval, the expected query weight on it is small. A hybrid argument then shows that the total simulation error is $1/\poly(n)$.

Next, we show another instance where queries to promise problems behave differently from languages. Note that the relationship $\BQP \subseteq \AWPP \subseteq \APP$~\cite{FR99,Fen03} also holds for the corresponding promise classes.
\begin{result}[\cref{lem:gapp_in_fp_promiseawpp,lem:ch_collapse}]\label{res:gapp_collapse}
    Under robust queries,
    \begin{align*}
        \GapP \subseteq \classname{FP}^{\PromiseAWPP}.
    \end{align*}
    Therefore,
    \begin{align*}
        \PromiseAWPP \subseteq \bqpqpoly \implies \CH = \YQPstar.
    \end{align*}
\end{result}
Note that $\AWPP$ is a natural counting upper bound for \BQP, and placing $\PromiseAWPP$ in $ \PromiseBQP$ would imply the collapse in \cref{res:gapp_collapse}, providing us evidence that the two promise classes differ. Furthermore, this improves an analogous result with $\classname{PromiseSBQP}$ in~\cite{Kup15}. More importantly for our purposes, Fenner showed that $\PP^{\APP} = \PP$~\cite{Fen03}, whereas no analogous theorem is known for $\classname{PromiseAPP}$. Since $\PromiseAWPP \subseteq \classname{PromiseAPP}$ and $\GapP \subseteq \classname{FP}^{\PromiseAWPP}$, such a result would collapse the counting hierarchy $\CH$.

The proof in~\cite{Yir25} that $\YQPstar$ is low for \PP relies on the containment $\YQPstar \subseteq \APP$. While this route does not work for \PromiseYQPstar, we nevertheless prove the same lowness conclusion. To do so, we describe an intermediate class \PostBQPstar, which is essentially \PostBQP where the postselected state depends on the input length, but not the input itself. Note that it lies in $\classname{PromiseAPP}$ due to the characterization of~\cite{MN16}. Adapting the proof-technique that $\PromiseBQP$ is low for $\PP$, we show that $\PP^\PostBQPstar = \PP$ and thus $\PP^{\PromiseYQPstar} = \PP$.

\subsection{Discussion}

The results we describe above support the simple idea that promise classes need to be considered carefully and separately. While many results transfer, others do not. Our results suggest that the treatment of off-promise inputs should be made explicit whenever semantic classes are used as oracles.

There are several questions which we leave open. First, can \cref{res:qcph} be improved to $\Pclass^\PP$ in order to match Toda's theorem? Second, for which other classes can we show an analogue of \cref{res:oracle_separation}, or are there classes for which this is impossible? Third, can loose oracle access to other classes be used to prove new results?

\paragraph{AI Disclosure.} We used Chat-GPT 5.6 Sol and Fable 5.1 to assist with exploring and assessing research directions, identifying potential errors, simplifying approaches and assist in producing \cref{fig:visualize_query_types}. The authors are solely responsible for the mathematical content and final presentation.

\section{Preliminaries}\label{sec:prelims}

For arbitrary positive integers $n,m$ such that $n \leq m$, we use $[n]$ to denote $\{1,\dots,n\}$ and $[n,m]$ for $\{n,n+1,\dots,m\}$. For a set $\calA$, we use $\Id_{\calA}(x) = 1$ if $x\in \calA$ and $0$ otherwise.

A promise problem $\Pi = (\Pi_\yes, \Pi_\no)$ is a pair such that $\Pi_\yes \cap \Pi_\no = \emptyset$ and $\Pi_\yes\cup \Pi_\no \subseteq \{0,1\}^*$. In general, we call the inputs in $\Pi_\yes\cup \Pi_\no$ the promise. A completion $L$ of $\Pi$ is a language such that $\Pi_\yes \subseteq L$ and $\Pi_\no \cap L = \emptyset$. The set of all such languages is denoted by $\comp(\Pi)$. As mentioned above, there are multiple ways to define queries to promise problems. The standard way is the one below.

\begin{definition}[Robust promise-queries]\label{def:robust_promise_queries}
    For a promise problems $\Pi$, a complexity class $C$, we say that the problem $\Gamma \in C^{\Pi}$ under robust promise-queries if,
    \begin{align*}
        \exists \text{an oracle machine $M$ of type $C$ } \forall A\in \comp(\Pi):\; M^A \text{ solves } \Gamma
    \end{align*}
\end{definition}
Access to a quantum machine to a completion $A$ is done using the XOR oracle $\ket{x,y} \rightarrow \ket{x, y\oplus \Id_A(x)}$. Unless indicated otherwise, the oracle access is robust. For complexity classes $C,D$, we write $C^D = \cup_{\Pi \in D} C^\Pi$.

\subsection{Complexity Zoo}

In this paper, we use a variety of complexity classes. Let us define the key classes discussed. For a comprehensive list, see~\cite{Zoo}. When considering a \emph{promise} version of a complexity class, we will use $\mathrm{Promise}$ at the start of its label, except for \QCPH as this is its conventional name. We are assuming that all quantum circuits are using Hadamard and Toffoli gates.

\begin{definition}
    Let $c,s: \mathbb{N} \to [0,1]$ be polynomial-time computable functions such that $c(n)-s(n) \geq \tfrac{1}{\poly(n)}$. A promise problem $\Pi$ is in $\PromiseBQP_{c,s}$ if there exists a uniform polynomial-time family of polynomial-sized quantum circuits $\{V_{\abs{x}}\}$ such that,
    \begin{align*}
        x\in \Pi_\yes &\implies \Pr[V_n(x) = 1] \geq c(\abs{x})\\
        x\in \Pi_\no &\implies \Pr[V_n(x) = 1] \leq s(\abs{x}).
    \end{align*}
    Under standard error-reduction,
    \begin{align*}
        \PromiseBQP \coloneq \bigcup_{c-s \geq \frac{1}{\poly(\abs{x})}} \PromiseBQP_{c,s}   = \PromiseBQP_{\frac{2}{3},\frac{1}{3}}.
    \end{align*}
\end{definition}

A standard $\PromiseBQP$-complete promise problem is quantum circuit evaluation, which we denote $\qcirc$. In \qcirc, the input is a description of a quantum circuit and the output is deciding whether it returns $1$ with probability at least $\tfrac{2}{3}$ or at most $\tfrac{1}{3}$. Without loss of generality we assume that every $\PromiseBQP$ oracle is represented using $\qcirc$. Next, we define the complexity class \QCPH.

\begin{definition}[Quantum-Classical Polynomial Hierarchy]\label{def:qcph} Set $\QCSigma{0} = \PromiseBQP$. For $i\geq 1$, a promise problem $\Pi$ is in $\QCSigma{i}$ if there exists a polynomial $p$ and a uniform \PromiseBQP verifier $V$ such that, with each $w_j \in \{0,1\}^{p(\abs{x})}$, 
    \begin{align*}
        x\in \Pi_\yes &\implies \exists w_1 \forall w_2 \dots Q_i w_i \text{ s.t. } \Pr[V(x,w_1,\dots, w_i) = 1] \geq \frac{2}{3},\\
        x\in \Pi_\no &\implies \forall w_1 \exists w_2 \dots \overline{Q}_i w_i \text{ s.t. } \Pr[V(x,w_1,\dots, w_i) = 1] \leq \frac{1}{3},
    \end{align*}
    where $Q_i$ is $\exists$ for odd $i$ and $\forall$ for even $i$, and $\overline{Q}_i$ is the complementary quantifier to $Q_i$. Set
    \begin{align*}
        \QCPH = \bigcup_{i\geq 0} \QCSigma{i}. 
    \end{align*}
\end{definition}

\begin{definition}[$\GapP, \PP$ and \AWPP]\label{def:gapP_pp}
    A function $g:\{0,1\}^* \rightarrow \mathbb{Z}$ is in $\GapP$ if there exists a deterministic polynomial-time branch verifier $M$ whose outputs are in $\{-1,0,1\}$ and a polynomial $p$ such that for all inputs $x$,
    \begin{align*}
        g(x) = \abs{\{w\in \{0,1\}^{p(\abs{x})}: M(x,w) = 1 \}} - \abs{\{w\in \{0,1\}^{p(\abs{x})}: M(x,w) = -1 \}}.
    \end{align*}
    Equivalently, $\GapP$ consists of differences of two $\classname{\# P}$ functions~\cite{FFK94}. A language $L$ is in $\PP$ if there exists $g\in \GapP$ such that,
    \begin{align*}
        x\in L \iff g(x) > 0.
    \end{align*}
    The counting hierarchy is defined as $\CH = \cup_i \CHi{i}$ where $\CHi{0} = \Pclass$ and $\CHi{i+1} = {\PP}^{\CHi{i}}$. 
    A language $L$ is in $\AWPP$ if there exists a $g\in \GapP$ and $q\in \poly$ such that,
    \begin{align*}
        x\in L &\implies \frac{2}{3} \leq \frac{g(x)}{2^{q(\abs{x})}} \leq 1,\\
        x\notin L &\implies 0 \leq \frac{g(x)}{2^{q(\abs{x})}} \leq \frac{1}{3}.
    \end{align*}
    A language $L$ is in $\APP$ if for every $r\in\poly$, there exist $g,d\in \GapP$ with $d(1^n) > 0$ such that,
    \begin{align*}
        x\in L &\implies 1 - 2^{-r(\abs{x})} \leq \frac{g(x)}{d(1^{\abs{x}})} \leq 1,\\
        x\notin L &\implies 0 \leq \frac{g(x)}{d(1^{\abs{x}})} \leq 2^{-r(\abs{x})}.
    \end{align*}
\end{definition}

\PromiseAWPP and \PromiseAPP are defined as \AWPP and \APP, except that the conditions are only imposed on $\Pi_\yes$ and $\Pi_\no$.
Note that access to an oracle $O$ is defined by allowing the non-deterministic machine to query $O$. The following closure properties will be useful.

\begin{lemma}[Closure properties of $\GapP$; In~\cite{FFK94}]\label{lem:closure_properties_gapp}
    For any $f,g\in \GapP$,
    \begin{itemize}
        \item $f+g$, $f-g$ and $f\cdot g$ are in $\GapP$.
        \item For a polynomial $p$ and $h(x,y) \in \GapP$,
        \begin{align*}
            \sum_{y\in \{0,1\}^{p(\abs{x})}} h(x,y) \in \GapP,\\
            \prod_{y\in [p(\abs{x})]} h(x,y) \in \GapP.
        \end{align*}
    \end{itemize}
    Furthermore, all statements above hold with respect to a fixed oracle $A$.
\end{lemma}

It has been shown that $\GapP$ can compute problems in $\PromiseBQP$.

\begin{lemma}[In~\cite{FR99}]\label{lem:bqp_gapP_representation}
    Let $\Pi \in \PromiseBQP$. There exists some nondecreasing polynomial $p$ and a function $H(x,1^k) \in \GapP$ such that $0\leq H(x,1^k) \leq 2^{p(\abs{x}+k)}$ and,
    \begin{align*}
        x\in \Pi_{\yes} &\implies \frac{H(x,1^k)}{2^{p(\abs{x}+k)}} \geq 1-2^{-k},\\
        x\in \Pi_{\no} &\implies \frac{H(x,1^k)}{2^{p(\abs{x}+k)}} \leq 2^{-k}.
    \end{align*}
\end{lemma}

Notice that \cref{lem:bqp_gapP_representation} immediately implies that $\PromiseBQP \subseteq \PromiseAWPP$, which relativizes. Next, let us define the advice classes we will discuss in the manuscript.

\begin{definition}[Advice classes]
    Let $\calC$ be one of $\PromiseP, \PromiseBPP$ or $\PromiseBQP$. The promise problem $\Pi$ is in $\calC_{\mathrm{/poly}}$ if there exist some $p\in \poly$, a function $h: \mathbb{N} \to \{0,1\}^*$ and a machine $M$ of type $\calC$ such that $\abs{h(n)} \leq p(n)$ and for all $x\in \Pi_\yes\cup \Pi_\no$, $M$ solves $\Pi$ on input $(x,h(\abs{x}))$.

    The class $\bqpqpoly$ is defined analogously, except $h$ outputs a quantum state on $p(n)$ qubits.
\end{definition}

The following result on advice will be useful.

\begin{theorem}[Theorem 1.1 in~\cite{Dru12}]\label{thm:distribution_solving}
    Let $f$ be a possibly partial Boolean function and $\mu$ a distribution supported on its domain such that any $T$-query randomized algorithm has success probability of at most $1-\epsilon$. For $k\geq 1$ and any randomized algorithm $M$ making at most $\epsilon T k$ queries,
    \begin{align*}
        \Pr[\text{$M$ solves $k$ independent instances of $f$}] \leq (2^{\epsilon}(1-\epsilon))^k.
    \end{align*}
\end{theorem}

Additionally, we will use the following promise class.

\begin{definition}[\PromiseYQPstar]\label{def:yqpstar}
    A promise problem $\Gamma$ is in \PromiseYQPstar if there exist two polynomial-time uniform quantum circuit families $\{A_n\}$ and $\{B_n\}$ such that,
    \begin{itemize}
        \item The \emph{advice} circuit $A_n$ inputs some state $\rho_0$ and produces some candidate advice state, outputs bit $b_{\mathrm{adv}} \in \{0,1\}$ and leaves a residual state $\rho_1$, conditioned on the observed value $b_{\mathrm{adv}}$. There exists some $\rho_0$ such that $\Pr[b_{\mathrm{adv}} = 1] \geq \tfrac{9}{10}$.  
        \item The \emph{evaluation} circuit $B_n$ takes in the input $x\in \{0,1\}^n$ and $\rho_1$, and outputs some bit $b_{\mathrm{out}}$. For any $\rho_0$ such that $\Pr[b_{\mathrm{adv}}=1] \geq \tfrac{1}{10}$ we have,
        \begin{align*}
            x\in \Gamma_\yes &\implies \Pr[b_{\mathrm{out}}=1| b_{\mathrm{adv}}=1] \geq \frac{9}{10},\\
            x\in \Gamma_\no &\implies \Pr[b_{\mathrm{out}}=1| b_{\mathrm{adv}}=1] \leq \frac{1}{10}.
        \end{align*}
    \end{itemize}
\end{definition}

We emphasize that $A_n$ in the definition above is independent of $x$, it only depends on the length $n$.

\begin{lemma}[Theorem 2.5 and Lemma 2.6 in~\cite{Yir25} based on~\cite{MW05}]\label{lem:yqpstar_postselection}
    Let $C$ be a polynomial-size quantum circuit and $r\in \poly(n)$. There is an amplified verifier $C^\prime$ which records bits $y,z$ such that for every eigenstate $\ket{\psi}$ of the acceptance operator $C$,
    \begin{align*}
        \Pr[C\ket{\psi} \text{ accepts}] \leq \frac{1}{10} &\implies \Pr[C^\prime\ket{\psi} \text{ accepts}] \leq 2^{r(n)},\\
        \Pr[C\ket{\psi} \text{ accepts}] \geq \frac{9}{10} &\implies \Pr[C^\prime\ket{\psi} \text{ accepts and } y=z=1] \geq (1-2^{r(n)})\frac{9}{20}.
    \end{align*}
\end{lemma}

\subsection{Loose oracle access}

Let us formally define \emph{loose} oracle access. Unless mentioned, we assume that $\Pi, \Theta$ and $\Gamma$ represent arbitrary promise problems.

\begin{definition}[Loose-access \BPP] A promise problem $\Pi$ is in $\LooseBPP^\Theta$ if there exists a polynomial $r$ and a probabilistic polynomial-time oracle machine $M$ such that,
    \begin{align*}
        x\in \Pi_{\yes} &\implies \abs{\left\{ s \in \{0,1\}^{r(\abs{x})}: \forall B \in \comp(\Theta),\;M^B(x,s) = 1 \right\} } \geq \frac{2}{3} 2^{r(\abs{x})}\\
        x\in \Pi_{\no} &\implies \abs{\left\{ s \in \{0,1\}^{r(\abs{x})}: \forall B \in \comp(\Theta),\;M^B(x,s) = 0 \right\} } \geq \frac{2}{3} 2^{r(\abs{x})}
    \end{align*}
\end{definition}

\begin{definition}[Randomized reduction]\label{def:bp_op} A promise problem $\Pi$ is in $\BPdot\Theta$ if there exists a polynomial $r$ and a polynomial-time computable function $R$ such that,
    \begin{align*}
        x\in \Pi_{\yes} &\implies \abs{\left\{ s \in \{0,1\}^{r(\abs{x})}: R(x,s) \in \Theta_\yes \right\} } \geq \frac{2}{3} 2^{r(\abs{x})}\\
        x\in \Pi_{\no} &\implies \abs{\left\{ s \in \{0,1\}^{r(\abs{x})}: R(x,s) \in \Theta_\no \right\} } \geq \frac{2}{3} 2^{r(\abs{x})}
    \end{align*}
    For a promise class $\calC$, $\BPdot\calC \coloneq \bigcup_{\Theta\in\calC} \BPdot\Theta$.
\end{definition}

In other words, $\Pi\in \BPdot\Theta$ if and only if $\Pi$ has a randomized polynomial-time many-one reduction to $\Theta$. As $R$ may be viewed as an oracle machine making a single query, $\BPdot\Theta \subseteq \LooseBPP^\Theta$.

\begin{definition}[Loose-access unambiguous polynomial time]\label{def:loose_access_up}
    A promise problem $\Pi$ is in $\LooseUP^\Theta$ if there exists some polynomial $r$ and a polynomial-time deterministic oracle machine $M$ such that,
    \begin{itemize}
        \item If $x\in \Pi_{\yes}$, there exists a unique witness $w^\prime \in \{0,1\}^{r(\abs{x})}$ such that for every $B\in \comp(\Theta)$, $M^B(x,w^\prime)=1$ and $M^B(x,w) = 0$ for every $w\in \{0,1\}^{r(\abs{x})} \setminus \{w^\prime\}$.
        \item If $x\in \Pi_{\no}$, then for all $w\in \{0,1\}^{r(\abs{x})}$ and all $B\in \comp(\Theta)$, $M^B(x,w)=0$.
    \end{itemize}
\end{definition}

We define the corresponding unambiguous hierarchy as follows.

\begin{definition}\label{def:uh}For a promise class $\calC$, the loose-access unambiguous polynomial hierarchy $\LooseUH^\calC$ is defined as follows,
    \begin{align*}
        \LooseUSigma{0}^{\calC} &= \calC\\
        \LooseUSigma{k+1}^{\calC} &= \bigcup_{\Gamma \in \LooseUSigma{k}^\calC} \LooseUP^{\Gamma}\\
        \LooseUH^\calC &= \bigcup_{k} \LooseUSigma{k}^{\calC}
    \end{align*}
\end{definition}

The following result will be useful.

\begin{lemma}[Lemmas 2 and 3 in~\cite{HS15}]\label{lem:loose_up_bpp}
    For any promise problem $\Theta$,
    \begin{align*}
        \LooseBPP^{\LooseBPP^\Theta} & = \LooseBPP^\Theta,\\
        \LooseUP^{\LooseBPP^\Theta} &\subseteq \LooseBPP^{\LooseUP^\Theta}.
    \end{align*}
\end{lemma}

\subsection{Cohen genericity}

For two sets of strings $\calA,\calB \subseteq \{0,1\}^*$, their join is defined as,
\begin{align*}
    \calA \oplus \calB \coloneq \{0x: x\in \calA\} \cup \{1x: x\in \calB\}.
\end{align*}
A \emph{Cohen condition} is a partial function $\sigma: \{0,1\}^* \rightharpoonup \{0,1\}$ with a finite domain. For two conditions $\sigma, \tau$ and a set $\calA$, we say,
\begin{align*}
    \sigma \preceq \tau &\iff \dom(\sigma) \subseteq \dom(\tau) \text{ and } \forall x\in \dom(\sigma), \tau(x)=\sigma(x),\\
    \sigma \prec \calA &\iff \forall x\in \dom(\sigma), \Id_{\calA}(x) = \sigma(x).
\end{align*}
A set $D$ of Cohen conditions is \emph{dense} and \emph{upwards closed} respectively when,
\begin{align*}
    \forall \sigma\, \exists \tau \in D&: \sigma \preceq \tau,\\
    \sigma \in D \text{ and } \sigma \preceq \tau &\implies \tau \in D.
\end{align*}

\begin{definition}[Arithmetic definability]
    Fix some encoding of finite Cohen conditions by natural numbers denoted $c$. A set $D$ of Cohen conditions is arithmetically definable relative to an oracle $B$ if there exists a first-order formula $\lambda(z)$ augmented with $B$ such that,
    \begin{align*}
        \sigma \in D \iff (\mathbb{N}, +, \times, 0,1, B) \vDash \lambda(c(\sigma)).
    \end{align*}
\end{definition}

\begin{definition}[Cohen genericity]
    For some $B\subseteq \{0,1\}^*$, the set $G\subseteq \{0,1\}^*$ is Cohen generic relative to $B$ if for all dense sets of Cohen conditions $D$,
    \begin{align*}
       \text{if $D$ is arithmetically definable relative to $B$} &\implies \exists \sigma \in D \text{ s.t. } \sigma \prec G
    \end{align*}
\end{definition}

\begin{lemma}\label{lem:enumeration_of_dense_sets}
    For some $B\subseteq \{0,1\}^*$, let $D_0, D_1,\dots,$ be an enumeration of all dense sets of Cohen conditions which are arithmetically definable relative to $B$. Then there exists an increasing sequence of finite Cohen conditions $\sigma_0 \preceq \sigma_1 \preceq \dots $ and a set $G\subseteq \{0,1\}^*$ such that for all $i\in \mathbb{N}$, $\sigma_{i+1} \in D_i$ and $\sigma_i \prec G$.
\end{lemma}
\begin{proof}
    Set $\sigma_0 = \emptyset$. Consider some $\sigma_i$. Using the density of $D_i$, choose some condition $\sigma_{i+1} \in D_i$ such that $\sigma_i \preceq \sigma_{i+1}$. This gives us an increasing sequence of finite conditions meeting every set in the enumeration. Define the set $G$ as,
    \begin{align*}
        G \coloneq \{x\in \{0,1\}^*: \exists i \text{ s.t. } x\in \dom(\sigma_i) \text{ and } \sigma_i(x)=1 \}.
    \end{align*}
    Notice that this gives us that $\sigma_i \prec G$ and $\sigma_{i+1} \in D_i$.
\end{proof}

% This is not the exact statement of 6.18, but it directly implies it.
\begin{theorem}[Consequence of Theorem 6.18 in \cite{FFKL03}]\label{thm:cohen_genericity_pspace}
    Let $B$ be a \PSPACE-complete oracle. Then if $G$ is Cohen generic relative to $B$,
    \begin{align*}
        \Pclass^{B\oplus G} = \BPP^{B\oplus G} = \BQP^{B\oplus G} = \AWPP^{B \oplus G}.
    \end{align*}
\end{theorem}

\section{Quantum-Classical Toda's theorem}\label{sec:qcph}

In the following section, let us show that $\QCPH \subseteq \PromiseBPP^\PP$. Most of the argument does not use quantum computation, so we fix an arbitrary function $f:\{0,1\}^*\to [0,1]$ throughout this section.

\begin{problem}[Quantified Threshold Problem]\label{def:quantified_threshold} Fix some $i\in \mathbb{N}_0$, a quantifier string $Q \in \{\exists, \forall \}^i$ and a polynomial-time computable verifier $V$ mapping $(x,w_1,\dots,w_i)$ to a string. Let $n=\abs{x}$, $l\in \poly(n)$ and assume that every witness $w_j \in \{0,1\}^{l(n)}$ for $j\in [i]$. For $j\leq i$, let
    \begin{align*}
        v_j(V, Q,x,w_1,\dots w_j) &\coloneq \begin{cases}
            f(V(x,w_1,\dots, w_i)) &\text{if $i=j$}\\
            \max_{w_{j+1}} v_{j+1}(V,Q,x,w_1,\dots, w_j, w_{j+1}) &\text{if $j <i$ and $Q_{j+1} = \exists$}\\
            \min_{w_{j+1}} v_{j+1}(V,Q,x,w_1,\dots, w_j, w_{j+1}) &\text{if $j <i$ and $Q_{j+1} = \forall$}
        \end{cases}\\
        v(V,Q,x) &\coloneq v_0(V,Q,x).
    \end{align*}
    The promise problem $\qthr(V,Q)$ takes inputs $(x,\alpha, \beta, 1^m)$, where $\alpha,\beta$ are rationals represented in binary, $0\leq \beta < \alpha \leq 1$, $m\geq 1$ and $\alpha-\beta \geq \tfrac{1}{m}$. It is defined as,
    \begin{align*}
        (x,\alpha, \beta, 1^m) \in \qthr(V,Q)_\yes &\iff v(V,Q,x) \geq \alpha\\
        (x,\alpha, \beta, 1^m) \in \qthr(V,Q)_\no &\iff v(V,Q,x) \leq \beta.
    \end{align*}
    When we need to make the dependence on $f$ explicit, we write $\qthr_f(V,Q)$ and $v^f_j$.
\end{problem}

When $i=0$ and the identity map $\Id$, we use $\thr(f) \coloneq \qthr_f(\Id,\emptyset)$ to denote the gapped threshold problem of $f$, which asks whether $f(z) \geq \alpha$ or $f(z) \leq \beta$. Note that $\qthr(V,\emptyset)$ reduces to $\thr(f)$ by replacing $x$ with $V(x)$.

Let $\qacc$ denote the function that maps the description of a quantum circuit to the probability it outputs $1$. The problem $\thr(\qacc)$ is polynomial-time equivalent to $\qcirc$, meaning that $\thr(\qacc) \in \PromiseBQP$. We immediately get the following characterization of \QCPH.

\begin{observation}\label{obs:qcph_redefinition} Let $f=\qacc$. Fix some $i\geq 1$ and let $Q = (\exists, \forall, \exists \dots Q_i) \in \{\exists, \forall\}^i$.
    A promise problem $\Pi$ is in $\QCSigma{i}$ if and only if there exists a verifier $V$ such that,
    \begin{align*}
        x\in \Pi_\yes \implies (x,\tfrac{2}{3}, \tfrac{1}{3}, 111)\in \qthr(V,Q)_\yes\\
        x\in \Pi_\no \implies (x,\tfrac{2}{3}, \tfrac{1}{3}, 111)\in \qthr(V,Q)_\no.
    \end{align*}
    When $i=0$, we assume that $Q= \emptyset$.
\end{observation}

\subsection{Promise-version of the first half of Toda's theorem}

First, let us show a promise-version of the first half of Toda's theorem. We will use the following standard helper statement.

\begin{lemma}\label{lem:nested_sets}
    Let $\emptyset \subset S \subseteq T \subseteq \{0,1\}^n$ such that $\abs{T}\leq 2 \abs{S}$. Choose a uniformly random $l\in [0,n+1]$ and a pairwise-independent hash function $h: \{0,1\}^n \to \{0,1\}^l$. Then,
    \begin{align*}
        \Pr_{l,h}[ \abs{T \cap h^{-1}(0^l)} = 1 \text{ and } T\cap h^{-1}(0^l)\subseteq S ] \geq \frac{1}{16(n+2)}
    \end{align*}
\end{lemma}
\begin{proof}
    Choose some $l^\prime$ such that $2\abs{T} \leq 2^{l^\prime} \leq 4\abs{T}$. Fix some $s\in S$. As $h$ is pairwise-independent,
    \begin{align*}
        \Pr_{h}[h(s) = 0^{l^\prime} \text{ and } \forall t\in T\setminus \{s\}, h(t)\neq 0^{l^\prime} ] &\geq 2^{-l^\prime} \left(1-\frac{\abs{T}-1}{2^{l^\prime}} \right)\\
        &\geq 2^{-l^\prime -1}.
    \end{align*}
    As the events for each $s$ are mutually exclusive, their probabilities add. Hence, the probability over all $s$ is at least $\tfrac{\abs{S}}{2^{l^\prime + 1}} \geq \tfrac{1}{16}$. However, this is conditioned on choosing $l^\prime$. Hence, the bound follows from the fact that $l$ is chosen randomly from $n+2$ values.
\end{proof}

\begin{lemma}\label{lem:threshold_safe_existential}
    Fix some verifier $V$ and a quantified string $Q$ which begins with $\exists$.
    For an input $(x,\alpha,\beta,1^r)$ satisfying the conditions in \cref{def:quantified_threshold}, let $m=m(\abs{x})$ be the length of the first witness.
    Letting $\delta = \tfrac{\alpha - \beta}{m+1}$, for $j\in [m+1]$ let $\alpha_j = \beta + j\delta$ and $\beta_j = \beta + (j-1)\delta$. Define the promise problem $\Theta$ as,
    \begin{align*}
        (x,\alpha,\beta,1^r,w,j) \in \Theta_\yes &\iff v_1(V,Q,x,w) \geq \alpha_j\\
        (x,\alpha,\beta,1^r,w,j) \in \Theta_\no &\iff v_1(V,Q,x,w) \leq \beta_j,
    \end{align*}
    and suppose that $\Theta \in \LooseBPP^C$ where $C$ is a promise class. Then the problem $\Gamma$ defined as,
    \begin{align*}
        (x,\alpha,\beta,1^r)\in \Gamma_\yes &\iff v(V,Q,x) \geq \alpha\\
        (x,\alpha,\beta,1^r)\in \Gamma_{\no} &\iff v(V,Q,x) \leq \beta,
    \end{align*}
    is in $\LooseBPP^{\LooseUP^C}$.
\end{lemma}
\begin{proof}
    For a fixed promised input $(x,\alpha,\beta,1^r)$ and $j\in[m+1]$, let $S_j = \{w: v_1(V,Q,x,w) \geq \alpha_j\}$, $T_j = \{w: v_1(V,Q,x,w) > \beta_j\}$ and $S_0 = \{w: v_1(V,Q,x,w) \geq \beta \}$. Notice that $S_j \subseteq T_j \subseteq S_{j-1} \dots \subseteq \{0,1\}^m$.

    Assume that $x\in \Gamma_\yes$. As $v(V,Q,x)\geq \alpha$, $S_{m+1}$ is non-empty. There is some $j\in[m+1]$ such that $\abs{T_j} \leq 2\abs{S_j}$ as if this was false, then $\abs{S_0} \geq 2^{m+1} \abs{S_{m+1}} \geq 2^{m+1}$, contradicting $S_0 \subseteq \{0,1\}^m$.

    Let us define the following promise problem $\Psi$. The input is a tuple $(x,\alpha,\beta, 1^r,j,l,h)$ where $l\in [0,m+1]$ and $h$ is a hash-function. An input is in $\Psi_\yes$ if and only if there exists a unique witness $w\in \{0,1\}^m$ such that $h(w)=0^l$ and $(x,\alpha,\beta, 1^r,w,j)\in \Theta_\yes$. On the other hand, the input is in $\Psi_\no$ if $(x,\alpha,\beta,1^r,w,j)\in \Theta_\no$ for all witnesses $w$. A verifier guesses $w$, checks whether $h(w)=0^l$, and accepts exactly when the $\Theta$ oracle answers yes. Therefore,
    \begin{align*}
        \Psi \in \LooseUP^{\Theta} \subseteq \LooseUP^{\LooseBPP^C}.
    \end{align*}
    Consider a randomized oracle machine $M$ which chooses a uniformly random $j\in [m+1]$ and chooses $l,h$ as in \cref{lem:nested_sets} and repeats this process polynomially many times, accepting if any run outputs $1$. Assuming $M$ guesses the index $j$ such that $\abs{T_j} \leq 2\abs{S_j}$, by \cref{lem:nested_sets} $M$ will choose a pair of $(l,h)$ which satisfy its conditions with probability at least $\tfrac{1}{16(m+2)}$. In this event, the query to $\Psi$ would be a valid yes instance. Hence, the probability of success is at least $\tfrac{1}{16(m+2)^2}$ per trial.

    Polynomial repetition raises this to any arbitrary constant less than $1$, let us say $\tfrac{2}{3}$. On the other hand, if $x\in \Gamma_\no$, we have that $v_1(V,Q,x,w)\leq \beta$ for all witnesses $w$. Hence, every query to $\Psi$ is a no instance, meaning that $M$ will always reject. By applying \cref{lem:loose_up_bpp}, we have that,
    \begin{align*}
        \Gamma \in \LooseBPP^{\LooseUP^{\LooseBPP^C}} &\subseteq \LooseBPP^{\LooseBPP^{\LooseUP^C}}\\
        &\subseteq \LooseBPP^{\LooseUP^C}.\qedhere
    \end{align*}
\end{proof}

\begin{lemma}\label{lem:qcph_in_bpp_uh}
    Fix $i\geq 0$, some verifier $V$ and a quantifier $Q \in \{\exists, \forall\}^i$. Then,
    \begin{align*}
        \qthr(V,Q)\in \LooseBPP^{\LooseUSigma{i}^{\thr(f)}}.
    \end{align*}
\end{lemma}
\begin{proof}
    We induct over $i$, simultaneously for all $f$. When $i=0$, $\qthr(V,\emptyset)$ can be efficiently transformed to an instance of $\thr(f)$ by replacing $x$ with $V(x)$, implying the statement.

    Assume the statement holds for $i-1$ and consider any instance of $\qthr(V,Q)$ with $\abs{Q}=i$ and input $(x,\alpha, \beta, 1^r)$ in the promise. Suppose that the first quantifier in $Q$ is $\exists$. Then we have that $v(V,Q,x) = \max_{w} v_1(V,Q,x,w)$.

    Let $Q^\prime = (Q_2,\dots ,Q_i)$. For every first witness $w$ and threshold index $j$, the question whether $v_1(V,Q,x,w) \geq \alpha_j$ or $v_1(V,Q,x,w) \leq \beta_j$ is an instance of \qthr with $Q^\prime$. By induction, the resulting promise problem is in $\LooseBPP^{\LooseUSigma{i-1}^{\thr(f)}}$.
    By \cref{lem:threshold_safe_existential},
    \begin{align*}
        \qthr(V,Q) \in \LooseBPP^{\LooseUSigma{i}^{\thr(f)}}.
    \end{align*}
    Now, suppose that the first quantifier is $\forall$. Let $\tilde{f} = 1-f$ and $\tilde{Q}$ be obtained by swapping $\exists$ and $\forall$. Notice that by \cref{def:quantified_threshold},
    \begin{align*}
        v^{\tilde{f}}(V,\tilde{Q},x) = 1 - v(V,Q,x).
    \end{align*}
    By replacing $(\alpha,\beta)$ with $(1-\beta, 1-\alpha)$, we can map yes instances of $\qthr(V,Q)$ to no instances of $\qthr_{\tilde{f}}(V, \tilde{Q})$ and vice-versa. Hence they are complements of each other. As $\qthr_{\tilde{f}}(V, \tilde{Q})$ has $\exists$ at the start of $\tilde{Q}$, it is in $\LooseBPP^{\LooseUSigma{i}^{\thr(\tilde{f})}}$.  As $\LooseBPP^C$ is closed under complement for all promise oracles $C$, the same applies to $\qthr(V,Q)$.
\end{proof}

\subsection{Exact arithmetization}

In the following statement, \emph{uniform} means fixing the nondeterministic machine before the oracle completions.

\begin{lemma}[Exact \SPP representation]\label{lem:exact_spp_repr}
    Fix a promise problem $\Pi$. For any $i\geq 0$ and a promise problem $\Theta \in \LooseUSigma{i}^{\Pi}$, there exists a uniform family of functions $g^{A} \in \GapP^{A}$ for all $A\in\comp(\Pi)$ such that,
    \begin{align}
        x\in \Theta_{\yes} &\implies g^{A}(x) = 1\label{eq:spp_implication_yes}\\
        x\in \Theta_{\no} &\implies g^{A}(x) = 0.\label{eq:spp_implication_no}
    \end{align}
\end{lemma}
\begin{proof}
    We show the statement by induction on $i$. In the base case when $i=0$, we let $g^A(x)$ be the indicator function on whether $x\in A$. This is in $\GapP^A$ as the non-deterministic machine may query whether $x\in A$ and satisfies the conditions as the inputs in the promise are fixed in $A$.

    Assume the claim holds for $i$ and let $\Theta \in \LooseUSigma{i+1}^{\Pi}$. By \cref{def:uh}, there exists some $\Gamma \in \LooseUSigma{i}^{\Pi}$ such that $\Theta \in \LooseUP^{\Gamma}$. Furthermore, let $h^A \in \GapP^A$ be the application of the inductive hypothesis, meaning,
    \begin{align}
        z\in \Gamma_\yes \implies h^A(z) = 1,& &z\in \Gamma_\no \implies h^A(z)=0.\label{eq:inductive_hypothesis}
    \end{align}
    Let $M$ be the verifier witnessing $\Theta \in \LooseUP^{\Gamma}$, $r$ some polynomial and $w \in \{0,1\}^{r(\abs{x})}$ its witness. By \cref{def:loose_access_up}, we have that for any $x\in \Theta_{\yes} \cup \Theta_{\no}$ and witness $w$, there exists some $c_x(w) \in \{0,1\}$ such that $M^B(x,w) = c_x(w)$ for all $B\in\comp(\Gamma)$ and moreover,
    \begin{align}
        \sum_{w\in \{0,1\}^{r(\abs{x})}} c_x(w) = \begin{cases}
            1 &\text{if $x\in \Theta_\yes$,}\\
            0 &\text{if $x\in \Theta_\no$.}\label{eq:condition_on_machines}
        \end{cases}
    \end{align}
    Let us abbreviate $u=(x,w)$. Without loss of generality, we may assume that every computation uses exactly $t=t(\abs{u})$ distinct queries, as we may pad and cache queries.

    Let $a\in \{0,1\}^t$ represent a transcript of query answers, $q_i(u, a_{<i})$ denote the $i$th query made based on the previous query outputs $a_{<i} = (a_1,\dots,a_{i-1})$, $\sigma(u,a) \in \{0,1\}$ the output of $M$ under the transcript $a$ and $\calQ_u$ denote the set of all possible query inputs. Lastly, we let $Z:\calQ_u \to \{0,1\}$ be a variable and define the transcript polynomial $P_u(Z)$ as,
    \begin{align}
        P_u(Z) &\coloneq \sum_{a \in \{0,1\}^t} \sigma(u,a)\prod_{i\in [t]} [a_i \cdot Z(q_i(u, a_{<i})) + (1-a_i)(1 - Z(q_i(u, a_{<i})))].\label{eq:def_transcript_polynomial}
    \end{align}
    For any $B\in \comp(\Gamma)$, when $Z$ agrees with $B$ we have that the product in \cref{eq:def_transcript_polynomial} equals $1$ when $a$ is the transcript in $B$ and is $0$ otherwise. Hence $P_u(B) = M^B(u)$. Lastly, as no queries are repeated, $P_u(Z)$ is a multilinear polynomial.
    
    Next we fix all the variables $q\in \Gamma_{\yes}\cup \Gamma_\no$ such that $Z(q) = \Gamma(q)$ and leave the rest free. Every Boolean assignment to the remaining off-promise variables corresponds to a completion $B$, meaning that the resulting multilinear polynomial equals $c_x(w)$ on all inputs. By uniqueness of multilinear representations on $\{0,1\}^n$~\cite{OD14}, the polynomial $P_u(Z)$ after fixing the promise-input variables is exactly represented as the constant $c_x(w)$. Therefore, we may substitute arbitrary integer values for the query result and still get the same output.

   Let us substitute $Z$ with $h^A$, obtaining,
   \begin{align*}
        H^A(u) &\coloneq \sum_{a\in \{0,1\}^t} \sigma(u,a) \prod_{i\in [t]}[a_i h^A(q_i(u, a_{<i}))  + (1-a_i) (1-h^A(q_i(u, a_{<i})))].
   \end{align*}
   By \cref{eq:inductive_hypothesis} and the discussion above on invalid queries, $h^A(q)$ behaves the same way as $Z$ conditioned on the promise. Hence for all $A\in\comp(\Pi)$,
   \begin{align}
     H^A(u) = c_x(w) \label{eq:equivalence_between_machines}.
   \end{align}
   Notice that as $q_j$ and $\sigma(u,a)$ are computable in polynomial time, $h^A(q_j(u,a_{<j})) \in \GapP^A$. By \cref{lem:closure_properties_gapp}, polynomially-long products and exponential sums are in $\GapP^A$ as well, meaning that $H^A(x,w)\in \GapP^A$.

   Finally, let $g^A(x) = \sum_{w} H^A(x,w)$. Again by \cref{lem:closure_properties_gapp}, $g^A \in \GapP^A$. Finally, \cref{eq:condition_on_machines,eq:equivalence_between_machines} imply \cref{eq:spp_implication_yes,eq:spp_implication_no}, completing the proof.
\end{proof}

Next, we compress the random seed of loose-access computation into the input, turning it into a randomized many-one reduction.

\begin{lemma}[Seed compression]\label{lem:seed_compression}
    Let $\Gamma \in \LooseBPP^\Theta$ via a machine $M$ using $r$ random bits, and suppose there is a uniform family of functions $h^A\in \GapP^A$ for all $A\in \comp(\Pi)$ which equals $1$ on $\Theta_\yes$ and $0$ on $\Theta_\no$. Define the promise problem $\Phi$ on pairs $(x,s)$ with $s\in \{0,1\}^{r(\abs{x})}$ as,
    \begin{align*}
        (x,s) \in \Phi_\yes &\iff x\in \Gamma_\yes \text{ and } \forall B\in \comp(\Theta),\; M^B(x,s) = 1,\\
        (x,s) \in \Phi_\no &\iff x\in \Gamma_\no \text{ and } \forall B\in \comp(\Theta),\; M^B(x,s) = 0.
    \end{align*}
    Then $\Gamma \in \BPdot\Phi$, and there is a uniform family of functions $g^A \in \GapP^A$ for all $A\in \comp(\Pi)$ which equals $1$ on $\Phi_\yes$ and $0$ on $\Phi_\no$.
\end{lemma}
\begin{proof}
    By the definition of $\LooseBPP^\Theta$, the map $R(x,s) = (x,s)$ is a randomized reduction from $\Gamma$ to $\Phi$. For a fixed $s$, $M(x,s)$ is a deterministic polynomial-time oracle machine whose output on $(x,s) \in \Phi_\yes \cup \Phi_\no$ is the same for every $B\in \comp(\Theta)$. Hence the argument in the proof of \cref{lem:exact_spp_repr} from \cref{eq:def_transcript_polynomial} to \cref{eq:equivalence_between_machines} applies verbatim with $u = (x,s)$ in place of $(x,w)$, yielding $g^A(x,s) \coloneq H^A(x,s) \in \GapP^A$ which equals $M^B(x,s)$ for all $B\in \comp(\Theta)$ and $A\in \comp(\Pi)$.
\end{proof}

Combining the results so far, we obtain the following statement for arbitrary $f$.

\begin{theorem}\label{thm:generic_qthr}
    Fix $i\geq 0$, some verifier $V$ and a quantifier $Q\in \{\exists,\forall\}^i$. Then there exists a promise problem $\Phi$ such that $\qthr(V,Q) \in \BPdot\Phi$, and a uniform family of functions $g^A \in \GapP^A$ for all $A\in \comp(\thr(f))$ which equals $1$ on $\Phi_\yes$ and $0$ on $\Phi_\no$.
\end{theorem}
\begin{proof}
    By \cref{lem:qcph_in_bpp_uh}, there exists some $\Theta\in \LooseUSigma{i}^{\thr(f)}$ such that $\qthr(V,Q) \in \LooseBPP^{\Theta}$. By \cref{lem:exact_spp_repr} with $\Pi = \thr(f)$, $\Theta$ has a uniform family $h^A\in \GapP^A$ which equals $1$ on $\Theta_\yes$ and $0$ on $\Theta_\no$. The statement follows from \cref{lem:seed_compression}.
\end{proof}

Next, we show how to use $\PP$ to solve problems with a $\PromiseBQP$ oracle.

\begin{lemma}\label{lem:bqp_completion_pp}
    Let $\Pi \in \PromiseBQP$ and $\Theta$ be a promise problem. Suppose that there is some nondeterministic polynomial-time oracle machine which for all $A\in \comp(\Pi)$ creates some $g^A \in \GapP^A$ such that,
    \begin{align*}
        x \in \Theta_{\yes} &\implies g^A(x) \geq 1\\
        x \in \Theta_{\no} &\implies g^A(x) \leq -1.
    \end{align*}
    Then $\exists C \in \comp(\Theta)$ such that $C\in \PP$.
\end{lemma}
\begin{proof}
    Fix some input $x$. By \cref{def:gapP_pp}, we may decompose $g$ as,
    \begin{align*}
        g^A(x) = \sum_{y\in \{0,1\}^{t(\abs{x})}} F^A(x,y),
    \end{align*}
    where $F^A(x,y) \in \{-1,0,1\}$ is a deterministic polynomial-time oracle branch verifier. Let $t=t(\abs{x})$. Without loss of generality as multiplying $g^A$ be a positive integer does not change its sign, we may assume that $t$ bounds the number and length of oracle queries made by $F^A(x,y)$. By caching to avoid repeated queries and padding with irrelevant queries, we assume that $t\geq 1$ and exactly $t$ distinct queries are made.

    Using the notation in \cref{lem:exact_spp_repr}, let $a_{<i} = (a_1,\dots,a_{i-1}) \in \{0,1\}^{i-1}$ be a set of query answers, $q_i(x,y,a_{<i})$ the $i$th query based on the previous query results which $F^A(x,y)$ makes and $\sigma(x,y,a)\in \{-1,0,1\}$ the final output using answers $a$.

    We define the set of all queries made as $\calQ_{x} = \{q_i(x,y,a_{<i}): y,a\in \{0,1\}^t, i\in [t]\}$. Immediately, we have that $\abs{\calQ_x} \leq t2^{2t}$. Letting $k=6t$, let $H\in \GapP$ and the nondecreasing polynomial $p\in \poly(\abs{x})$ be as in \cref{lem:bqp_gapP_representation} for $\Pi$. Last, we define the following function,
    \begin{align*}
        P_x(q) &\coloneq \frac{H(q,1^k)}{2^{p(\abs{q}+k)}}.
    \end{align*}
    Let us define the distribution $\mu$ over oracle assignment functions $\calO: \calQ_x \to \{0,1\}$ as the product distribution where, independently for each $q\in \calQ_x$, $\calO(q)=1$ with probability $P_x(q)$. Let $E_x$ denote the event that every $q\in \calQ_x \cap (\Pi_\yes \cup \Pi_\no)$ is answered correctly by $\calO$. Using the union bound,
    \begin{align*}
        \Pr_\mu[\neg E_x] &\leq \abs{\calQ_x} 2^{-6t}\\
        &\leq 2^{-3t} \eqcolon \delta.
    \end{align*}
    When $E_x$ occurs, $\calO$ is a completion of $\Pi$. Hence,
    \begin{align*}
        x \in \Theta_{\yes} &\implies g^{\calO}(x) \geq 1\\
        x \in \Theta_{\no} &\implies g^{\calO}(x) \leq -1.
    \end{align*}
    As $\abs{g^{\calO}(x)} \leq 2^t$ for all functions $\calO$,
    \begin{align*}
        x \in \Theta_{\yes} &\implies \E_\mu[g^{\calO}(x)] \geq 1 - \delta(2^t +1) > 0\\
        x \in \Theta_{\no} &\implies \E_\mu[g^{\calO}(x)] \leq -1 + \delta(2^t + 1) < 0.
    \end{align*}
    Notice that we may write $g^\calO$ as follows,
    \begin{align*}
        g^{\calO}(x) = \sum_{y,a\in \{0,1\}^t} \sigma(x,y,a) \prod_{i\in [t]} [a_i \calO(q_i(x,y,a_{<i})) + (1-a_i)(1-\calO(q_i(x,y,a_{<i})))],
    \end{align*}
    as the product only equals $1$ if $a$ agrees with $\calO$ and $0$ otherwise and hence the sum over all $a$ represents the result of a single $F^{\calO}(x,y)$. Notice that in a fixed transcript the queries made are distinct. Therefore, the sampled answers based on $\calO$ are independent. Using the linearity of expectation,
    \begin{align*}
        G(x) &\coloneq 2^{p(t+k)t} \E_{\mu}[g^{\calO}(x)]\\
        &= 2^{p(t+k)t} \sum_{y,a\in \{0,1\}^t} \sigma(x,y,a) \prod_{i\in [t]} [a_i P_x(q_i(x,y,a_{<i})) + (1-a_i)(1-P_x(q_i(x,y,a_{<i})))]\\
        &= \sum_{y,a\in \{0,1\}^t} \sigma(x,y,a) \prod_{i\in [t]} [a_i \hat{H}_x(q_i(x,y,a_{<i})) + (1-a_i)(2^{p(t+k)} - \hat{H}_x(q_i(x,y,a_{<i})))],
    \end{align*}
    where we let $\hat{H}_x(q) = 2^{p(t+k) - p(\abs{q}+k)}H(q,1^k)$. As $\abs{q}\leq t$ and $p$ is nondecreasing, the exponent is nonnegative, and $P_x(q) = \tfrac{\hat{H}_x(q)}{2^{p(t+k)}}$. Both $\hat{H}_x(q)$ and $2^{p(t+k)} - \hat{H}_x(q)$ are in $\GapP$.
    Every factor is a \GapP function of $x,y,a$, meaning that by \cref{lem:closure_properties_gapp} $G\in \GapP$. Therefore, we have that $C_{\Theta} \coloneq \{z: G(z) > 0\}$ is in $\PP$. As $\Theta_\yes \subseteq C_{\Theta}$ and $\Theta_{\no} \cap C_{\Theta} = \emptyset$, it is a valid completion of $\Theta$.
\end{proof}

Notice the immediate result from \cref{lem:bqp_completion_pp}.

\begin{corollary} Under robust promise-queries,
    \begin{align*}
        \PP^{\PromiseBQP} = \PP.
    \end{align*}
\end{corollary}

Finally, we may combine all the results above together.

\begin{corollary}[Quantum-Classical Toda's Theorem]\label{thm:qcph_bpp_pp}
    \begin{align*}
        \QCPH \subseteq \BPdot\PP \subseteq \PromiseBPP^\PP.
    \end{align*}
\end{corollary}
\begin{proof}
    Let $f = \qacc$ and $\Pi \coloneq \thr(\qacc) \in \PromiseBQP$. Consider some promise problem $\Gamma \in \QCPH$. By \cref{obs:qcph_redefinition}, we can consider some verifier $V$ and quantifier $Q\in \{\exists,\forall\}^i$. By \cref{thm:generic_qthr}, there exists some promise problem $\Phi$ such that $\Gamma \in \BPdot\Phi$ and a uniform family $h^A\in \GapP^A$ for $A\in \comp(\Pi)$ which equals $1$ on $\Phi_\yes$ and $0$ on $\Phi_\no$. Therefore $g^A = 2h^A - 1$ satisfies the conditions of \cref{lem:bqp_completion_pp}, so $\Phi$ has a completion $C\in \PP$. As the randomized reduction from $\Gamma$ to $\Phi$ maps yes instances into $\Phi_\yes \subseteq C$ and no instances into $\Phi_\no \subseteq \overline{C}$ with probability at least $\tfrac{2}{3}$, it is also a randomized reduction from $\Gamma$ to $C$, meaning $\Gamma \in \BPdot\PP$. The second inclusion was noted after \cref{def:bp_op}.
\end{proof}

\section{Promise problems and languages in relativized worlds}\label{sec:promise_separation}

In this section, we show that there exists a promise problem separating randomized and quantum machines in the relativized world where all of their languages are equal. To do so, we will encode instances of Simon's problem into a Cohen generic oracle.

For the remainder of the section, we will let $\calT$ represent the \tqbf language. 

\subsection{Simon's problem}

\begin{problem}\label{prob:simons}
    For $n\in \mathbb{N}$, let $\simyes$ denote the set of permutations over $\{0,1\}^n$ and $\simno$ denote the set of functions for which there exists some $s\in \{0,1\}^n \setminus \{0^n\}$ such that,
    \begin{align*}
        \forall x,y\in \{0,1\}^n, f(x) = f(y) \iff y\in \{x,x\oplus s\}.
    \end{align*}
    Simon's problem is to decide, given the promise that $f\in \simyes \cup \simno$, which set contains $f$.

    We use $\mu_n$ to denote the distribution where with probability $\tfrac{1}{2}$ we choose a uniformly random $f\in \simyes$ and otherwise a uniformly random $f\in \simno$.
\end{problem}

We will use the following textbook result.

\begin{theorem}[In~\cite{Sim97,BH97}]\label{thm:simons_result}
    The following hold about \cref{prob:simons}:
    \begin{enumerate}
        \item There exists an exact $\poly(n)$-query quantum algorithm solving \cref{prob:simons}.
        \item Every randomized query algorithm $R$ making at most $2^{n/4}$ queries satisfies,
        \begin{align*}
            \Pr_{f \sim \mu_n}[R^f\text{ answers correctly}] \leq \frac{1}{2} + 2^{-n/2}.
        \end{align*}
    \end{enumerate}
\end{theorem}

For the rest of the section, let $n\in \mathbb{N}$. Given inputs $x,u\in \{0,1\}^n$ and $j\in [n]$, we let
\begin{align*}
    \addr(x,u,j) = 1^n 0\, x\, u\, \mathrm{bin}(j),
\end{align*}
where $\mathrm{bin}$ is the binary representation of $j$ under a fixed length of $\ceil{\log(n+1)}$. In this problem, $x$ represents an instance of \cref{prob:simons}, $u$ represents an input to the instance and $j$ indexes an output bit. Notice that each instance of the problem $x$ has a unique encoding, which is always $\poly(n)$.

The set of all instances of size $n$ is denoted by $\calB_n \coloneq \{\addr(x,u,j): x,u\in \{0,1\}^n, j\in [n]\}$. For some arbitrary $G\subseteq \{0,1\}^*$, we define the problem instance $f_x^G: \{0,1\}^n \to \{0,1\}^n$ as 
\begin{align*}
    (f_x^G(u))_j = \Id_G(\addr(x,u,j)).
\end{align*}
Therefore, at length $n$, each $f_x^G$ describes the instance of \cref{prob:simons} in $G$. For a Cohen condition $\sigma$, if $\calB_n \subseteq \dom(\sigma)$, we define $f_x^{\sigma}$ analogously.

Let $A = \calT \oplus G$. We define the Simon's Promise Problem $\Gamma^A$ as,
\begin{align}
    x\in \Gamma^A_\yes &\iff f_x^G \in \simyes,\label{eq:simon_problem_yes}\\
    x\in \Gamma^A_\no &\iff f_x^G \in \simno.\label{eq:simon_problem_no}
\end{align}

We have that this problem is always in $\PromiseBQP^A$.

\begin{observation}
    For any $G \subseteq \{0,1\}^*$,
    \begin{align*}
        \Gamma^{\calT\oplus G} \in \PromiseBQP^{\calT \oplus G}.
    \end{align*}
\end{observation}
\begin{proof}
    Given an input $x\in \{0,1\}^*$, we may run the \BQP machine $M$ in \cref{thm:simons_result}. Notice that any query $(f_x^{G}(u))_j$ may be made by querying $1 \, \addr(x,u,j)$, which may be prepared in $O(n)$ steps. Assuming that $x\in \Gamma_\yes^{\calT \oplus G} \cup \Gamma_\no^{\calT \oplus G}$, the machine is always correct.
\end{proof}

\subsection{Density requirements}

Let $\calM_i = (M_i, t_i, a_i)$ be an enumeration over all deterministic machines $M_i$, a polynomial bound on their running time $t_i$ and a polynomial bound on the advice size $a_i$. Without loss of generality we may assume that after $t_i(\abs{x})$ steps, $M_i(x, a)$ outputs $0$ unless it has already halted.

Letting $b\in \{0,1\}$, for an arbitrary Cohen condition $\sigma$ we use $M_i^{\calT \oplus \sigma}(x,a)= b$ to denote the result of the oracle machine with advice $a\in \{0,1\}^{a_i(\abs{x})}$ when all queries $0\, w$ return $\calT(w)$ and $1\, w$ are in $\dom(\sigma)$ and return $\sigma(w)$. Therefore, the result $b$ is preserved under any extension of $\sigma$.

For some fixed enumeration $\calM_i$, we let $E_i$ be the set of Cohen conditions $\sigma$ for which there exists some $n\geq 12$ such that:
\begin{enumerate}
    \item $\calB_n \subseteq \dom(\sigma)$,
    \item for all $x\in \{0,1\}^n$, $f_x^\sigma \in \simyes \cup \simno$, and
    \item for all $a \in \{0,1\}^{a_i(n)}$, there exists some $x\in \{0,1\}^n$ and $b\in \{0,1\}$ such that 
    \begin{align*}
        M_i^{\calT \oplus \sigma}(x,a) = b \text{ and } b\neq \Id_{\simyes}(f_x^{\sigma}).
    \end{align*}
\end{enumerate}

\begin{lemma}\label{lem:ei_arithmetically_definable}
    For every $i\in \mathbb{N}$, $E_i$ is upward closed and arithmetically definable relative to $\calT$.
\end{lemma}
\begin{proof}
    Fix some $i$. Let us first prove upward closure. Let $\sigma \in E_i$ with $n$ and suppose that $\sigma \preceq \tau$ where $\tau$ is some Cohen condition. As $\calB_n \subseteq \dom(\sigma)$, the extension $\tau$ agrees with $\sigma$. Therefore, for all $x\in \{0,1\}^n$, $f_x^\tau = f_x^\sigma$, so the first two conditions in the definition of $E_i$ hold.

    For the last condition of upward closure, fix some advice string $a$. Therefore, there are $x$ and $b$ such that $M^{\calT \oplus \sigma}_i(x,a)= b$ and $b\neq \Id_{\simyes}(f_x^{\sigma})$. As $M$ only makes queries in $\dom(\sigma)$, its answers are not changed by $\tau$ and so $M^{\calT \oplus \tau}_i(x,a)$ also answers incorrectly. Hence $\tau \in E_i$.

    It remains to prove definability. Encode a finite Cohen condition by listing the finite pairs $(w,\sigma(w))$. Membership can be checked by checking whether there exists an $n$ which satisfies the three conditions above. Given a fixed $n$ and $\sigma$, each condition can be checked by finite searches using queries to $\calT$ and $\sigma$. Hence the formula is arithmetical relative to $\calT$.
\end{proof}

\begin{lemma}\label{lem:ei_dense}
    For every $i\in \mathbb{N}$, $E_i$ is dense.
\end{lemma}
\begin{proof}
    Fix some arbitrary $i$ and Cohen condition $\tau$. Let $l = \max\{ \abs{w}: w\in \dom(\tau)\} $ when $\tau \neq \emptyset$ and $l=-1$ otherwise. Furthermore, we define the constants $T_n \coloneq \floor{2^{n/4}}$ and $\rho \coloneq 2^{29/60}\tfrac{31}{60}$.
    Choose some $n\geq 12$ such that,
    \begin{align}
        \abs{\addr(x,u,j)} &> l\label{eq:minimum_length_query}\\
        t_i(n) &< \frac{29}{60}T_n\\
        2^{a_i(n)} \rho^{2^n} &< 1.
    \end{align}
    We know that some $n$ exists as both $t_i$ and $a_i$ are in $\poly(n)$ and $\rho < 1$. The bound in \cref{eq:minimum_length_query} means that $\calB_n \cap \dom(\tau) = \emptyset$.

    Extend $\tau$ to a finite condition $\tau^\prime$ by assigning $0$ to every query $M_i$ can make that are outside of $\calB_n$ while preserving all values already set by $\tau$. This fixes queries $w$ whose size is bounded by $\abs{w}\leq t_i(n)$, meaning that the fixture is indeed finite.

    Fix some advice string $a\in \{0,1\}^{a_i(n)}$. Let $\mathbf{f} = (f_x)_{x\in \{0,1\}^n}$ be a tuple of instances of \cref{prob:simons} and write its representation into $\tau^\prime$. Let $R_a(\mathbf{f})$ be the sequence of results of $M_i(x,a)$ over all $x\in \{0,1\}^n$. As all queries $M_i$ could make are fixed, the outputs are fixed as well. Notice that over all inputs the machines make, their total query count is bounded by,
    \begin{align*}
        2^n t_i(n) \leq \frac{29}{60} T_n 2^n.
    \end{align*}
    Suppose that we sample $f_x\sim \mu_n$ independently for all $x$. As $n\geq 12$, any randomized algorithm making at most $T_n$ queries succeeds with probability at most $\tfrac{31}{60}$. Applying \cref{thm:distribution_solving} with $\epsilon = \tfrac{29}{60}$ and $k=2^n$,
    \begin{align*}
        \Pr_{f_x\sim \mu_n}[R_a(\mathbf{f}) = (\Id_{\simyes}(f_x))_{x\in \{0,1\}^n}] \leq \rho^{2^n}
    \end{align*}
    As there are $2^{a_i(n)}$ advice strings, union-bounding gives us,
    \begin{align*}
        \Pr_{f_x\sim \mu_n}[\exists a \in \{0,1\}^{a_i(n)}: R_a(\mathbf{f}) = (\Id_{\simyes}(f_x))_{x\in \{0,1\}^n}] < 2^{a_i(n)}\rho^{2^n} < 1.
    \end{align*}
    As $\mu_n$ is entirely supported on $\simyes \cup \simno$, there must exist some tuple $\mathbf{f}^*$ for which there does not exist an advice string $a$ which solves all $2^n$ instances correctly. Therefore, if we extend $\tau^\prime$ by $\mathbf{f}^*$ to obtain $\sigma$, we find that $\sigma$ is a valid Cohen condition such that $\sigma \in E_i$. As $\tau$ was arbitrary, $E_i$ is dense.
\end{proof}

We prove the main result of the section.

\begin{theorem}\label{thm:promise_language_difference}
    There exists an oracle $A$ such that,
    \begin{align*}
        \Pclass^A = \BPP^A = \BQP^A = \AWPP^A,\\
        \PromiseBQP^A \not\subseteq \PPoly^A, \\
        \PromiseBPP^A \subseteq \PPoly^A.
    \end{align*}
\end{theorem}
\begin{proof}
    As the proof that $\BPP \subseteq \classname{P}_{\mathrm{/poly}}$~\cite{Adl78,BG81} relativizes and applies to promise problems, it suffices to construct some promise problem in $\PromiseBQP^A$ that is not in $\PPoly^A$.

    Let us build the oracle $G$. Fix some enumeration $\calD = D_0, D_1,\dots,$ of all dense sets which are arithmetically definable relative to $\calT$. Apply \cref{lem:enumeration_of_dense_sets} on $\calD$ in order to obtain the increasing sequence $(\sigma_i)$ and the oracle $G$. Let $A = \calT \oplus G$. As $G$ is Cohen generic relative to $\calT$, \cref{thm:cohen_genericity_pspace} gives us,
    \begin{align*}
        \Pclass^A = \BPP^A = \BQP^A = \AWPP^A.
    \end{align*}
    Let $\Gamma^A$ be the promise problem from \cref{eq:simon_problem_yes,eq:simon_problem_no}, meaning that $\Gamma^A \in \PromiseBQP^A$.

    Let us show that $\Gamma^A \notin \PPoly^A$. Consider the enumeration of machines $\calM_i$. By \cref{lem:ei_arithmetically_definable,lem:ei_dense}, every $E_i$ is dense and arithmetically definable relative to $\calT$. Hence, for each $i\in \mathbb{N}$, there exists some index $s_i$ such that $D_{s_i} = E_i$. By \cref{lem:enumeration_of_dense_sets}, $\sigma_{s_i+1} \in E_i$. Furthermore, as $E_i$ is upward closed, every $\sigma_s$ for $s\geq s_{i}+1$ remains in $E_i$.

    Let $n_i$ be the input size which witnesses that $\sigma_{s_i+1} \in E_i$. Hence $\sigma_{s_i+1}$ fixes the block $\calB_{n_i}$ and $\sigma_{s_i+1} \prec G$. Therefore, for all $x\in \{0,1\}^{n_i}$, $f_x^{\sigma_{s_i+1}} = f_x^G$. Note that the queries the machine $M_i$ makes outside of $f_x^G$ do not affect the output either, as they are fixed by $\sigma_{s_i+1}$ as well.

    Therefore, for every advice string $a\in \{0,1\}^{a_i(n_i)}$, there exists some input $x\in \{0,1\}^{n_i}$ on which $M_i^A$ answers incorrectly. As this applies to all $\calM_i$, there cannot be a $\calM_i$ which solves the problem. Hence,
    \begin{align*}
        &\Gamma^A \notin \PPoly^A.\qedhere
    \end{align*}
\end{proof}

% \PromiseBQP and friends
\section{\BQP and its variants under promises}\label{sec:bqp_variants}

\subsection{Self-lowness}

We establish a promise-version of the result that \BQP is self-low.

\begin{theorem}\label{thm:promise_bqp_selflow}
    Under robust promise-queries,
    \begin{align*}
        \PromiseBQP^\PromiseBQP = \PromiseBQP.
    \end{align*}
\end{theorem}
\begin{proof}
    The containment from right to left is trivial, so let us focus on the reverse containment. Let $\Pi \in \PromiseBQP$ and $\Gamma \in \PromiseBQP^\Pi$. Hence, by \cref{def:robust_promise_queries}, there is a uniform circuit family $\{Q_x^A\}$ whose descriptions are independent of $A$ which correctly decides $\Gamma$ for every completion $A\in\comp(\Pi)$ whose error is at most $\tfrac{1}{16}$. Let $q\in \poly(\abs{x})\geq 1$ be the number of oracle calls in $Q_x^A$ and $L\in \poly(\abs{x})$ the size of the input register. We may assume that each query has size $L$.
    Let $V$ be the circuit for $\Pi$ and let $p_z = \Pr[V(z) = 1]$.

    Next, we partition the off-promise interval where $p_z\in (\tfrac{1}{3}, \tfrac{2}{3})$. Let $\rho \in \mathbb{N} \geq 2$ be a power of $2$. For $j\in [\rho]$, define
    \begin{align*}
        &\alpha_j \coloneq \frac{1}{3} + \frac{j}{3\rho}, &\beta_j \coloneq \frac{1}{3} + \frac{j-1}{3\rho}.
    \end{align*}
    Hence the intervals $[\beta_j, \alpha_j)$ partition $[\tfrac{1}{3}, \frac{2}{3})$. Notice that for all queries $z\in \{0,1\}^L$, for all $j\in [\rho]$ except at most one, either $p_z \geq \alpha_j$ or $p_z \leq \beta_j$. If there exists a $j$ such that $\beta_j < p_z < \alpha_j$, we call such $j$ ambiguous for $z$.
    Last, when $j$ is not ambiguous for $z$, we define the corresponding output bit $b_j(z)$ as,
    \begin{align*}
        b_j(z) &\coloneq \begin{cases}
            1 &\text{if $p_z \geq \alpha_j$,}\\
            0 &\text{if $p_z \leq \beta_j$.}
        \end{cases}
    \end{align*}
    Let $h:\{0,1\}^L \to [\rho]$ be a random function we will specify later. Suppose that we want to simulate the output $b_{h(z)}(z)$ to some queried input $z$. Let $\epsilon > 0$ and consider the simulation of the query by making $O(\rho^2 \log(1/ \epsilon))$ repetitions of $V(z)$ and comparing their mean with $\tfrac{\alpha_{h(z)} + \beta_{h(z)}}{2}$. 
    
    If $h(z)$ is not ambiguous for $z$, the decision error is at most $\epsilon$. By running the simulation of the query, computing the decision bit into the output register and then uncomputing the procedure, we are able to simulate the query $b_{h(z)}(z)$. We denote this operation by $\tilde{O}_h$. For nonambiguous $z$, on the subspace with input $z$, the Euclidean norm between the real oracle and simulation is at most $2\sqrt{\epsilon}$. Let $\tilde{Q}_x^h$ be the modification of $Q_x^{A}$ where each oracle uses $\tilde{O}_h$ and let $a_x(h)$ denote its acceptance probability when using $h$.

    Next, we analyze $\tilde{Q}_x^h$ for a random $h$ where each input is uniformly and independently assigned to a value in $[\rho]$. Let us define two helper functions $F, S: \{0,1\}^L \to \{0,1\}$. $S$ will indicate whether we hit the ambiguous range using $h$, while $F$ denotes the output using the procedure above. When $z$ does not have an ambiguous index $j$, we let $F(z) = b_{h(z)}(z)$ and $S(z) = 0$. Notice that this holds for all $z\in \Pi_\yes \cup \Pi_\no$.

    In the other case, let $r \in [\rho]$ be the ambiguous index for $z$ and set $S(z) = \Id_{\{r\}}(h(z))$. When $S(z) = 0$, let $F(z) = b_{h(z)}(z)$. Otherwise, independently for each $z$, choose $F(z)$ randomly as,
    \begin{align*}
        \Pr[F(z) = 1| h(z) = r] = \frac{r-1}{\rho - 1}.
    \end{align*}
    Therefore,
    \begin{align*}
        \Pr_h[F(z) = 1 | S(z) = 0] = \frac{r-1}{\rho -1} = \Pr_h[F(z) = 0 | S(z) = 1],
    \end{align*}
    where the probabilities include auxiliary randomness independent of $h$.
    Therefore, $F$ and $S$ are independent. Additionally, notice that $\Pr[S(z) = 1] \leq \tfrac{1}{\rho}$ and if we extend $F$ to strings of arbitrary length such that it agrees with the promise inputs, then $F$ forms a completion of $\Pi$. We will use $O_F$ to denote the application of the oracle using the completion $F$.

    For each $t\in [q]$, write the state which is running $Q_x^F$ immediately before its $t$th query $\ket{\psi_t^F}$ and its ambiguous query weight $\wt_t(F,S)$ as,
    \begin{align*}
        \ket{\psi_t^F} &= \sum_{z} \ket{z} \ket{\phi_{t,z}^F},\\
        \wt_t(F,S) &= \sum_{z: S(z) = 1} \norm{\ket{\phi_{t,z}^F}}^2.
    \end{align*}
    On the query subspace where $S(z)=0$, we have that the norm between a real oracle call and a simulation is at most $2\sqrt{\epsilon}$. Therefore,
    \begin{align*}
        \norm{(\tilde{O}_h - O_F)\ket{\psi_t^F} } \leq 2\sqrt{\epsilon} + 2\sqrt{\wt_t(F,S)}.
    \end{align*}
    Condition on some arbitrary $F$, meaning that the state $\ket{\psi_t^F}$ is fixed. As $F$ and $S$ are independent,
    \begin{align*}
        \E[\wt_t(F,S)|F] &= \sum_{z}\norm{\ket{\phi_{t,z}^F}}^2 \Pr[S(z) = 1|F] \leq \frac{1}{\rho}.
    \end{align*}
    By Jensen's inequality, $\E[\sqrt{\wt_t(F,S)}] \leq \tfrac{1}{\sqrt{\rho}}$. By applying the triangle inequality and taking the expectation over the choice of $h$ and the auxiliary randomness in $F$,
    \begin{align}
        \E\left[\norm{\ket{\tilde{\psi}_q^h} - \ket{\psi_q^F}}\right] &\leq 2q\sqrt{\epsilon} + \frac{2q}{\sqrt{\rho}}.\label{eq:expectation_states}
    \end{align}
    Let $\epsilon \coloneq (4096q^2)^{-1}$ and $\rho$ be the least power of two which is at least $4096q^2$. Then, we have that \cref{eq:expectation_states} is at most $\tfrac{1}{16}$. For normalized pure states, the difference in acceptance probabilities is at most their Euclidean distance. Therefore,
    \begin{align*}
        x\in \Gamma_\yes &\implies \E_h[a_x(h)] \geq \frac{7}{8},\\
        x\in \Gamma_\no &\implies \E_h[a_x(h)] \leq \frac{1}{8}.
    \end{align*}
    Lastly, note that until now we assumed that $h$ is a uniformly random function. However, as $h$ is exponential, we cannot just randomly instantiate it. Instead, notice that by polynomial method, $a_x(h)$ can be written as a multilinear polynomial of degree $2q$~\cite{BBC+01}. By standard analysis of limited-independence constructions, there exists some seed $s\in \poly(\abs{x})$ which initializes a family of functions $\{h_s\}$ which are $2q$-independent and uniform over $[\rho]$~\cite{Vad12}. Combining the two together,
    \begin{align*}
        \E_s[a_x(h_s)] = \E_h[a_x(h)].
    \end{align*}
    The final \PromiseBQP computation is run by initializing a random seed $s$ in superposition, which represents some $h_s$ and then running $\tilde{Q}_x^{h_s}$. As the seeds form an orthgonal basis, we have that the probability this circuit accepts is equal to $\E_{s}[a_x(h_s)]$, completing the proof.
\end{proof}

A similar approach also works for advice classes.

\begin{corollary}\label{cor:promise_bqpqpoly_selflow}
    Under robust promise-queries,
    \begin{align*}
        \bqpqpoly^{\bqpqpoly} = \bqpqpoly
    \end{align*}
\end{corollary}
\begin{proof}
    As we are using trusted advice, we may arbitrarily ask for polynomial number of independent copies of the advice in order to perform any repetitions of simulations. Therefore, the same proof as the one for \cref{thm:promise_bqp_selflow} applies.
\end{proof}

\subsection{\PromiseAWPP versus \PromiseBQP}

\begin{lemma}\label{lem:gapp_in_fp_promiseawpp}
    $\GapP \subseteq \classname{FP}^{\PromiseAWPP}$.
\end{lemma}
\begin{proof}
    Fix some $f\in \GapP$. We define the promise problem $\Pi^f$ whose inputs are the triples $q = (x,L,U)$ where $L<U$ are binary-encoded integers as,
    \begin{align*}
        q \in \Pi^f_\yes &\iff L \leq f(x) \leq U \text{ and } f(x) \geq U - \floor*{\frac{U - L}{4}}\\
        q\in \Pi^f_\no &\iff L \leq f(x) \leq U \text{ and } f(x) \leq L + \floor*{\frac{U - L}{4}}
    \end{align*}
    Let us show that $\Pi^f \in \PromiseAWPP$. Let $n=\abs{q}$, $p(n) \coloneq 2n+4$ and,
    \begin{align*}
        g(q) \coloneq \floor*{\frac{2^{p(n)}}{U - L}} (f(x) - L).
    \end{align*}
    On bad incodings or inputs with $U\leq L$, let $g(q)=0$ and place those inputs outside of the promise.
    Notice that $p(n)\in \poly(n)$ so by \cref{lem:closure_properties_gapp}, $g(q) \in \GapP$. Suppose that $q\in \Pi^f_\yes$. Then,
    \begin{align*}
        \frac{g(q)}{2^{p(n)}} &\geq \frac{f(x) - L}{2^{p(n)}} \floor*{\frac{2^{p(n)}}{U - L}}\geq \frac{f(x) - L}{U - L} - \frac{f(x) - L}{2^{p(n)}}\\
        &\geq \frac{3}{4} - \frac{U - L}{2^{p(n)}} > \frac{2}{3},
    \end{align*}
    where the last line used the promise and that $1 \leq U - L < 2^{n+1}$, which means that $\tfrac{U-L}{2^{p(n)}} \leq \tfrac{1}{16}$. Similarly, when $q\in \Pi^f_\no$, we have $\tfrac{g(q)}{2^{p(n)}} < \tfrac{1}{3}$, so $\Pi^f\in \PromiseAWPP$.

    Let us show how to compute $f(x)$ using an arbitrary fixed completion $A\in \comp(\Pi^f)$. Fix some nondeterministic polynomial-time machine whose gap is $f$ and let $b\in \poly$ such that $\abs{f(x)} \leq 2^{b(\abs{x})}$. Consider the following protocol.

    \begin{algbox}
        \begin{enumerate}
            \item Let $L_0 = -2^{b(\abs{x})}, U_0 = 2^{b(\abs{x})}$
            \item Repeat starting with $i=0$:
            \begin{enumerate}
                \item If $U_i = L_i$, end the loop.
                \item Otherwise, query $q = (x,L_i, U_i)$ to get $a_i = A(q)$ and let $s_i = \floor{(U_i - L_i)/4}$.
                \item If $a_i = 1$, then let $L_{i+1} = L_i + s_i + 1$ and $U_{i+1} = U_i$.
                \item Otherwise, let $L_{i+1} = L_i$ and $U_{i+1} = U_i - s_i - 1$.
                \item Increment $i$ by $1$.
            \end{enumerate}
            \item Output $U_i$.
        \end{enumerate}
    \end{algbox}
    Notice that at every step $i$, $L_i \leq f(x) \leq U_i$. When the oracle answers $1$, the bottom quarter cannot contain $f(x)$ due to the promise. Similarly if the output is $0$, the top quarter cannot contain $f(x)$. This remains true when we are querying off-promise, because either answer preserves $f(x)$ when it lies off-promise. The interval decreases by a constant fraction, meaning the loop termintates in polynomially many steps. Hence we are able to compute $f(x)$ using $O(b(\abs{x}))$ queries to $\Pi^f$, meaning that $f\in \classname{FP}^{\Pi^f}$.
\end{proof}

% This is subsumed below.
% \begin{corollary}
%     If $\PromiseAWPP = \PromiseBQP$, then the counting hierarchy colllapses to \PP
% \end{corollary}
% \begin{proof}
%     Directly from \cref{thm:promise_bqp_selflow,lem:gapp_in_fp_promiseawpp}.
% \end{proof}

\begin{lemma}\label{lem:ch_collapse}
    If $\PromiseAWPP \subseteq \bqpqpoly$, then $\CH = \YQPstar$.
\end{lemma}
\begin{proof}
    By \cref{lem:gapp_in_fp_promiseawpp,cor:promise_bqpqpoly_selflow}, the hypothesis implies that $\PP \subseteq \bqpqpoly$. As $\PP$ is a syntactic class, we have $\PP \subseteq \classname{BQP_{/qpoly}}$ and thus \cite[Theorem 3.3]{Yir25} implies $\CH = \YQPstar$.
\end{proof}

\subsection{Stronger variants of \BQP}

Next, we define a variant of \PostBQP, where the postselection state-preparetion circuit depends only on the length of the input. As we show below, it functions as an intermediate class between $\PromiseYQPstar$ and $\bqpqpoly$ and lends itself well to the techniques discussed above.

\begin{definition}\label{def:postbqp_star}
    A promise problem $\Gamma$ is in $\PostBQPstar$ if there exist two polynomial-time uniform quantum circuit families $\{P_n\}$ and $\{D_n\}$ such that,
    \begin{itemize}
        \item The \emph{postselection} circuit $P_n$ produces some state $\ket{\psi}$ with a postselection bit $p$ such that $\Pr[p=1] > 0$. Let $\rho_n$ be the normalized state $\ket{\psi}$ after conditioning on $p=1$.
        \item The \emph{evaluation} circuit $D_n$ takes as input $x\in \{0,1\}^n$ together with $\rho_n$ and satisfies,
        \begin{align*}
            x\in \Gamma_\yes &\implies \Pr[D_n(x,\rho_n) = 1| p=1] \geq \frac{2}{3},\\
            x\in \Gamma_\no &\implies \Pr[D_n(x,\rho_n) = 1| p=1] \leq \frac{1}{3}.
        \end{align*}
    \end{itemize}
\end{definition}

We emphasize that $P_n$ does not take the input $x$.

\begin{proposition}\label{prop:yqp_start_in_postbqp_star}
    $\PromiseYQPstar \subseteq \PostBQPstar \subseteq \classname{PromiseAPP}$.
\end{proposition}
\begin{proof}
    Let us prove the first containment. Fix some $\Gamma\in\PromiseYQPstar$, witnessed by $A_n,B_n$ and some input $x\in \{0,1\}^n$. Let $m$ be the size of the advice and $\{\ket{\psi_i}\}_{i\in [2^m]}$ the eigenbasis of the acceptance operator of $A_n$ and $p_i$ the probability that $A$ accepts $\ket{\psi_i}$. As some advice state is accepted with probability at least $\tfrac{9}{10}$, there exists some $i^\prime$ such that $p_{i^\prime} \geq \tfrac{9}{10}$.

    Apply \cref{lem:yqpstar_postselection} to $A_n$ with $r = m+ 10$, obtaining $A^\prime_n$. We define the postselection circuit $P_n$ as follows. It prepares the maximally mixed state, runs $A^\prime_n$ and postselects on $A_n^\prime$ accepting and $y=z=1$. Afterwards, $P_n$ runs $A_n$ again and postselects on $b_{\mathrm{adv}} = 1$, outputting the final state.

    For each $i$, let $q_i \coloneq \Pr[A^\prime_n\ket{\psi_i} \text{ accepts and }y=z=1]$. As the initial state is the maximally mixed state, the postselection probability is $Z=2^{-m} \sum_i q_i p_i$ which is strictly positive due to the index $i^\prime$. Let $\calB = \{i: p_i < \tfrac{1}{10}\}$ and $\delta \coloneq \tfrac{2^{-m}\sum_{i\in \calB} q_i p_i}{Z}$ to be the weight in the postselected state of branches for which the correctness of $B_n$ does not have to hold. By \cref{lem:yqpstar_postselection}, $q_i \leq 2^{-r}$ for $i\in \calB$. Therefore,
    \begin{align*}
        \delta &\leq \frac{2^{-m} \sum_{i\in \calB} q_i p_i}{2^{-m} \sum_i q_i p_i} \leq \frac{2^{-m} 2^m 2^{-r}}{2^{-m} p_i q_i}\\
        & \leq \frac{2^{-r}}{2^{-m} (1-2^{-r}) \tfrac{81}{200}} < \frac{1}{100}.
    \end{align*}
    If we condition on the postselected event, every branch outside of the one covered by $\delta$ has $p_i \geq \tfrac{1}{10}$, which is the condition on $\rho_0$ and $B_n$ in \cref{def:yqpstar}. Therefore, letting the evaluation circuit $D_n=B_n$,
    \begin{align*}
        x\in \Gamma_\yes &\implies \Pr[D_n(x, \rho_n)=1] \geq (1-\delta) \frac{9}{10} > \frac{2}{3},\\
        x\in \Gamma_\no &\implies \Pr[D_n(x, \rho_n)=1] \leq (1-\delta) \frac{1}{10} + \delta < \frac{1}{3}.
    \end{align*}
    Therefore, $P_n$ and $D_n$ satisfy the conditions in \cref{def:postbqp_star}, meaning that $\Gamma \in \PostBQPstar$.

    The second follows from the characterization of $\classname{APP}$ as $\PostBQP$ where the postselection event can be written using some $d\in \GapP$ which depends only on the input size~\cite{MN16}. As the acceptance probability of any circuit using Hadamard and Toffoli gates can be exactly written using a $\GapP$ functions~\cite{FR99}, the inclusion follows.
\end{proof}

We observe the following oracle separation.

\begin{proposition}
    There exists a classical oracle $O$ such that
    \begin{align*}
        \PromiseYQPstar^O \subsetneq \PostBQPstar^O \subsetneq \bqpqpoly^O.
    \end{align*} 
\end{proposition}
\begin{proof}
    Note that all inclusions hold for any oracle $O$ as \cref{prop:yqp_start_in_postbqp_star} relativizes and $\bqpqpoly$ can simulate $\PostBQPstar$ by using the post-selected state $\rho$ as the advice state.

    It was shown in~\cite[Corollary 20]{AKKT20} that there exists a classical oracle $O$ and unary languages $L_0,L_1 \in \SBP^O$ such that one of $L_0,L_1$ is not in $\QMA^O$. As $\SBP \subseteq \PP$ and $\PP = \PostBQP$~\cite{Aar05} relativize, there exist postselection algorithms which solve $L_0$ and $L_1$. As the languages are unary, we can make both algorithms independent of the input, hence placing it in \PostBQPstar. On non-unary strings, the circuit automatically rejects. On the other hand, as $\YQPstar \subseteq \QMA$, one of $L_0,L_1$ is not in $\PromiseYQPstar$, completing the first separation.

    For the second separation, we keep the same oracle $O$. There are countably many uniform $\PostBQPstar^O$ machines and hence only countably many unary languages in it. On the other hand, $\bqpqpoly^O$ contains every unary language, of which there are uncountably many.
\end{proof}

\begin{lemma}\label{lemma_postbqp_star_low}
    $\PP^{\PostBQPstar} = \PP$.
\end{lemma}
\begin{proof}
    Fix some $\Pi\in \PostBQPstar$ and $L\in \PP^\Pi$. For simplicity, we use the notation of \cref{lem:bqp_completion_pp}. There is a uniform family of functions $g^A\in \GapP^A$ which can be decomposed using $F^A$ where $A\in \comp(\Pi)$. Similarly to before, we assume that each every computation makes exactly $t\in \poly(\abs{x})$ distinct queries of length $t$. Let $q_i(x,y,a_{<i})$ be the $i$th query after answers $a_{<i}$ and let $\sigma(x,y,a)\in \{-1,0,1\}$ be the final output. Let $\calQ_x$ be the set of all possible queries, whose size is bounded by $t2^{2t}$.
    
    Let $P_n$ and $D_n$ be the circuits for $\Pi$ from \cref{def:postbqp_star}. By \cref{lem:bqp_gapP_representation}, there exist functions $d(1^l), h(q) \in \GapP$ and $p\in \poly(l)$ such that for every query $q$ such that $\abs{q}=l$,
    \begin{align*}
        &\Pr[P_l \text{ postselects}] = \frac{d(1^l)}{2^{p(l)}}, &\Pr[P_l \text{ postselects and } D_l(q) = 1] = \frac{h(q)}{2^{p(l)}}.
    \end{align*}
    Furthermore, $1\leq d(1^l)$, $0 \leq h(q) \leq d(1^l)$ and,
    \begin{align*}
        q\in \Pi_\yes &\implies \frac{h(q)}{d(1^l)} \geq \frac{2}{3},\\
        q\in \Pi_\no &\implies \frac{h(q)}{d(1^l)} \leq \frac{1}{3}.
    \end{align*}
    Let $\epsilon = 2^{-6t}$ and $N=O(t)$ be an odd integer such that taking the majority over $N$ trials reduces the error to $\epsilon$~\cite{Hoe63}. For $\abs{q} \leq t$, let
    \begin{align*}
        H(x,q) \coloneq \sum_{\substack{z\in \{0,1\}^N \\ \sum_i z_i > N/2}} \prod_{i\in [N]} (z_i h(q) + (1 - z_i)(d(1^{\abs{q}}) - h(q))).
    \end{align*}
    and $H(x,q) = 0$ when $\abs{q} > t$. By \cref{lem:closure_properties_gapp}, $H(x,q) \in \GapP$. Furthermore, if we let $P_x(q) \coloneq \tfrac{H(x,q)}{d(1^{\abs{q}})^N}$, we have that $P_x(q)$ solves $\Pi$ up to error $\epsilon$. Additionally, we may ensure that $P(x,q)$ has a common denominator for all $\abs{q}\leq t$ by multiplying both the numerator and denominator by $\prod_{l: \abs{q}\neq l} d(1^{l})^N$. Let $D_x \coloneq \prod_{l} d(1^{l})^N$ be this denominator and $\hat{H}(x,q) \coloneq H(x,q) \prod_{l: \abs{q}\neq l} d(1^{l})^N$. Notice that $P_x(q) = \tfrac{\hat{H}(x,q)}{D_x}$ and by \cref{lem:closure_properties_gapp} both $\hat{H}(x,q)$ and $D_x$ are in $\GapP$.

    The rest of the proof essentially follows the proof of \cref{lem:bqp_completion_pp}. Let $\mu_x$ be the product distribution on assignments $\calO: \calQ_x \to \{0,1\}$ where $\Pr[\calO(q) = 1] = P_x(q)$. By union-bounding, we have that the probability that $\calO$ does not agree with $\Pi$ on $\calQ_x \cap (\Pi_\yes \cup \Pi_\no)$ is at most $\delta \coloneq 2^{-3t}$.

    Hence with probability at least $1-\delta$, $\calO$ is a valid completion of $\Pi$. Therefore,
    \begin{align*}
        x\in L &\implies \E_{\mu_x}[g^\calO(x)] \geq 1 -\delta(2^t +1) > 0\\
        x\notin L &\implies \E_{\mu_x}[g^\calO(x)] \leq -1 +\delta(2^t +1) < 0.
    \end{align*}
    It remains to write this expectation as the sign of a $\GapP$ function. Applying the independence of queried strings and linearity of expectation,
    \begin{align*}
        G(x) &\coloneq D_x^t \E_{\mu_x}[g^\calO(x)]\\
            &= \sum_{y,a\in \{0,1\}^t} \sigma(x,y,a) \prod_{i\in [t]}\left[ a_i\hat{H}(x,q_i(x,y,a_{<i})) + (1-a_i)(D_x - \hat{H}(x,q_i(x,y,a_{<i})))  \right].
    \end{align*}
    By \cref{lem:closure_properties_gapp}, $G(x)\in \GapP$. As $D_x^t > 0$, we have that $x\in L \iff G(x) > 0$ and hence $L\in \PP$.
\end{proof}

Immediately from \cref{prop:yqp_start_in_postbqp_star,lemma_postbqp_star_low}, we have the following.

\begin{corollary}
    $\PP^{\PromiseYQPstar} = \PP$.
\end{corollary}

\bibliographystyle{alpha}
\bibliography{main}

\end{document}